\documentclass[a4paper,UKenglish]{lipics-v2018}

\usepackage{microtype}%
\usepackage{tikz}
\usetikzlibrary{calc}
\usetikzlibrary{decorations.pathreplacing}

\usepackage[linesnumbered,ruled]{algorithm2e}

\usepackage{intcalc}

\usepackage{todonotes}

\title{Parallel and I/O-efficient 
    Randomisation of Massive Networks
  using Global Curveball Trades}

\titlerunning{Randomisation of Massive Networks using Global Curveball Trades}%

\author{Corrie Jacobien Carstens}{University of Amsterdam, Netherlands}{c.j.carstens@uva.nl}{}{}{}
\author{Michael Hamann}{Karlsruhe Institute of Technology, Germany}{michael.hamann@kit.edu}{}{}{}
\author{Ulrich Meyer}{Goethe University, Frankfurt, Germany}{umeyer@ae.cs.uni-frankfurt.de}{}{}{}
\author{Manuel Penschuck}{Goethe University, Frankfurt, Germany}{mpenschuck@ae.cs.uni-frankfurt.de}{}{}{}
\author{Hung Tran}{Goethe University, Frankfurt, Germany}{htran@ae.cs.uni-frankfurt.de}{}{}{}
\author{Dorothea Wagner}{Karlsruhe Institute of Technology, Germany}{dorothea.wagner@kit.edu}{}{}

\authorrunning{C. J. Carstens, M. Hamann, U. Meyer, M. Penschuck, H. Tran and D. Wagner}%

\Copyright{Corrie Jacobien Carstens, Michael Hamann, Ulrich Meyer, Manuel Penschuck, Hung Tran and Dorothea Wagner}%

\subjclass{%
  Mathematics of computing
  $\rightarrow$ Random graphs}


\category{}%

\relatedversion{}%

\supplement{%
    Stable versions of \imcb{} and \emgcb{} are released as part of NetworKit (\url{http://network-analysis.info}). 
}

\funding{This work was partially supported by Deutsche Forschungsgemeinschaft (DFG) under grants ME~2088/3-2, ME~2088/4-2, and WA~654/22-2.}

\ArticleNo{11} %
\nolinenumbers %
\hideLIPIcs  %

\theoremstyle{bold}
\newtheorem{mydef}[theorem]{Definition}
\newtheorem{mylem}[theorem]{Lemma}
\newtheorem{myexp}[theorem]{Example}
\newtheorem{myprop}[theorem]{Proposition}
\newtheorem{myrem}[theorem]{Remark}
\newtheorem{mycor}[theorem]{Corollary}
\newtheorem{mythm}[theorem]{Theorem}

\DeclareMathOperator{\avg}{avg}
\DeclareMathOperator{\sort}{sort}
\DeclareMathOperator{\scan}{scan}
\DeclareMathOperator{\rank}{rank}

\newcommand{\newalgoname}[3]{%
    \expandafter\gdef\csname #1\endcsname{\hyperref[#2]{\textsl{#3}}}
    \expandafter\gdef\csname #1Head\endcsname{\texorpdfstring{\textsl{#3}}{#3}}
}    

\newalgoname{emhh}{sec:intro}{EM-HH}
\newalgoname{emes}{sec:intro}{EM-ES}
\newalgoname{emlfr}{sec:intro}{EM-LFR}
\newalgoname{vles}{sec:experiments}{VL-ES}
\newalgoname{imcb}{subsec:algo-imcb}{IM-CB}
\newalgoname{emcb}{subsec:algo-emcb}{EM-CB}
\newalgoname{emgcb}{subsec:algo-emgcb}{EM-GCB}
\newalgoname{empgcb}{subsec:algo-empgcb}{EM-PGCB}
\newalgoname{tfp}{subsec:tfp}{TFP}
\newalgoname{esmc}{subsection:es}{ESMC}

\newcommand{\sequence}[3]{\ensuremath{[\,#1\,]_{#2}^{#3}}}

\newcommand{\A}{\mathcal A}
\newcommand{\Nei}[1]{\ensuremath{\A_{#1}}}
\newcommand{\AG}{\mathscr{A}_G}
\newcommand{\Suc}{\ensuremath{\mathscr{S}}}
\renewcommand{\O}{\mathcal O} %
\newcommand{\pld}[3]{\ensuremath{\textsc{Pld}\,([#1, #2), #3)}}
\newcommand{\defrel}{\ensuremath{:=}}

\newcommand{\hide}[1]{}

\newenvironment{myalgorithm}[1][htb]
{
  \let\oldnl\nl%
	\newcommand{\nonl}{\renewcommand{\nl}{\let\nl\oldnl}}%
	
  \SetCommentSty{myalgorithmCommSty}
  
	\DontPrintSemicolon
	\SetAlgoLined\SetAlgoNoEnd
	\SetArgSty{normalfont}
	\SetKwInOut{Input}{Input}
	\SetKwInOut{Output}{Output}
	\SetKw{Continue}{continue}
	\SetKwFor{ParForEach}{par foreach}{}{}
	\SetKw{Sequentially}{sequentially}
	\begin{algorithm}[#1]\small%
	}{\end{algorithm}}
 \makeatletter%
\begin{document}

\maketitle

\begin{abstract}
  Graph randomisation is a crucial task in the analysis and synthesis of networks.
  It is typically implemented as an \emph{edge switching} process (\emph{ESMC}) repeatedly swapping the nodes of random edge pairs while maintaining the degrees involved~\cite{Mihail2003}.
  Curveball is a novel approach that instead considers the whole neighbourhoods of randomly drawn node pairs.
  Its Markov chain converges to a uniform distribution, and experiments suggest that it requires less steps than the established \emph{ESMC} \cite{Carstens2016}.
  Since trades however are more expensive, we study Curveball's practical runtime by introducing the first efficient Curveball algorithms:
  the I/O-efficient \emph{EM-CB} for simple undirected graphs and its internal memory pendant \emph{IM-CB}.
  Further, we investigate \emph{global trades}~\cite{Carstens2016} processing every node in a graph during a single super step, and show that undirected global trades converge to a uniform distribution and perform superior in practice.
  We then discuss \emph{EM-GCB} and \emph{EM-PGCB} for global trades and give experimental evidence that \emph{EM-PGCB} achieves the quality of the state-of-the-art \emph{ESMC} algorithm \emph{EM-ES} \cite{Hamann2017a} nearly one order of magnitude faster. 
\end{abstract}

\clearpage

\section{Introduction}\label{sec:intro}
In the analysis of complex networks, such as social networks, the underlying graphs are commonly compared to random graph models to understand their structure~\cite{Itzkovitz2003,Newman2001,Strogatz2001}.
While simple models like Erd\H{o}s-R\'enyi graphs~\cite{er-rg-59} are easy to generate and analyse, they are too different from commonly observed powerlaw degree sequences~\cite{Newman2001,n-tsfcn-03,Strogatz2001}.
Thus, random graphs with the same degree sequence as the given graph are frequently used~\cite{Cobb2003,Itzkovitz2003,Schlauch2015Motif}.
In practice, many of these graphs are simple graphs, i.e. graphs without self-loops and multiple edges.
In order to obtain reliable results in these cases, the graphs sampled need to be simple since non-simple models can lead to significantly different results~\cite{Schlauch2015, Schlauch2015Motif}.
The randomisation of a given graph is commonly implemented as an edge switching Markov chain \esmc{}~\cite{Cobb2003,Milo2003}.

Nowadays, massive graphs that cannot be processed in the RAM of a single computer, require new analysis algorithms to handle these huge datasets. 
In turn, large benchmark graphs are required to evaluate the algorithms' scalability --- in terms of speed and quality.
LFR is a standard benchmark for evaluating clustering algorithms which repeatedly generates highly biased graphs that are then randomised~\cite{Lancichinetti2009,Lancichinetti2008}.
\cite{Hamann2017a} presents the external memory LFR generator \emlfr{} and its I/O-efficient edge switching \emes{}.
Although \emes{} is faster than previous results even for graphs fitting into RAM, it dominates \emlfr{}'s running time.
Alternative sampling via the Configuration Model~\cite{Molloy1995} was studied to reduce the initial bias and the number of \esmc{} steps necessary~\cite{Hamann2017b}.
Still, graph randomisation remains a major bottleneck during the generation of these huge graphs.

The Curveball algorithm has been originally proposed for randomising binary matrices while preserving row and column sums~\cite{Strona2014, Verhelst2008}
and has been adopted for graphs~\cite{CarstensPhd, Carstens2016}:
instead of switching a pair of edges as in \esmc{}, Curveball trades the neighbours of two nodes in each step.
Carstens et~al. further propose the concept of a \emph{global trade}, a super step composed of single trades targetting every node\footnote{For an odd number~$n$ of nodes, a single node is left out} in a graph once~\cite{Carstens2016}.
The authors show that global trades in bipartite or directed graphs converge to a uniform distribution, and give experimental evidence that global trades require fewer Markov-chain steps than single trades.
However, while fewer steps are needed, the trades themselves are computationally more expensive.
Since we are not aware of previous efficient Curveball algorithms and implementations, we investigate this trade-off here.

\textbf{Our contributions}.
We present the first efficient algorithms for Curveball: the (sequential) internal memory and external memory algorithms \imcb{}\footnote{%
    We prefix internal memory algorithms with \texttt{IM} and I/O-efficient algorithms with \texttt{EM}.
    The suffices \texttt{CB}, \texttt{GCB}, and \texttt{PGCB} denote Curveball, CB. with global trades, and parallel CB. with global trades respectively.}
and \emcb{} for the Simple Undirected Curveball algorithm (see \autoref{sec:algo}).
Experiments in \autoref{sec:experiments}, indicate that they are faster than the established edge switching approaches in practice.

In \autoref{section:randschemes}, we show that random global trades lead to uniform samples of simple, undirected graphs and demonstrate experimentally in \autoref{sec:experiments} that they converge even faster than the corresponding number of uniform single trades.
Exploiting structural properties of global trades, we simplify \emcb{} yielding \emgcb{} and the parallel I/O-efficient \empgcb{} which achieves \emes{}'s quality nearly one order of magnitude faster in practice (see \autoref{sec:experiments}).

\section{Preliminaries and Notation}\label{sec:notation}
We define the short-hand $[k] \defrel \{1, \ldots, k\}$ for $k \in \mathbb N_{>0}$, and write $\sequence{x_i}{i=a}{b}$ for an ordered sequence $[x_a, x_{a+1}, \ldots, x_b]$.

\textbf{Graphs and degree sequences.}
A graph $G = (V, E)$ has $n=|V|$ sequentially numbered nodes $V = \{v_1, \ldots, v_n\}$ and $m=|E|$ edges.
Unless stated differently, graphs are assumed to be undirected and unweighted.
To obtain a unique representation of an \emph{undirected} edge $\{u, v\} \in E$, we use \emph{ordered} edges $[u, v] \in E$ implying $u \le v$; in contrast to a directed edge, the ordering is used algorithmically but does not carry any meaning.
A graph is called \emph{simple} if it contains neither multi-edges nor self-loops, i.e. $E \subseteq \{\, \{u,v\}\, |\, u,v \in V \text{ with } u \ne v\, \}$.
For node $u \in V$ define the \emph{neighbourhood} $\Nei{u} := \{ v : \{u,v\} \in E \}$ and \emph{degree} $\deg(u) \defrel |\Nei{u}|$.
Let $d_{\max} \defrel \max_v\{\deg(v)\}$ be the maximal degree of a graph.
A vector $\mathcal D = \sequence{d_i}{i=1}{n}$ is a degree sequence of graph $G$ iff $\forall v_i \in V\colon \deg(v_i) = d_i$.

\textbf{Randomisation and Distributions.} $\pld ab\gamma$ refers to an integer \underline{P}ower\underline{l}aw \underline{D}istribution with exponent $-\gamma \in \mathbb R$ for $\gamma \ge 1$ and values from the interval $[a, b)$; 
let $X$ be an integer random variable drawn from $\pld ab\gamma$ then $\mathbb P[X{=}k] \propto k^{-\gamma}$ (proportional to) if $a \le k < b$ and $\mathbb P[X{=}k] = 0$ otherwise.
A statement depending on some number $x > 0$ is said to hold \emph{with high probability} if it is satisfied with probability at least $1 - 1/x^c$ for some constant $c \ge 1$.
Let $S$ be a finite set, $x \in S$ and let $\sigma$ be permutation on $S$, we define $\rank_\sigma(x)$ as the number of elements positioned in front of $x$ by $\sigma$.

\subsection{External-Memory Model}\label{ssec:emm}
In contrast to classic models of computation, such as the unit-cost random-access machine, modern computers contain deep memory hierarchies ranging from fast registers, over caches and main memory to solid state 
drives (SSDs) and hard disks. 
Algorithms unaware of these properties may face performance penalties of several orders of magnitude.

We use the commonly accepted external memory (EM) model by Aggarwal and Vitter~\cite{Aggarwal1988} to reason about the influence of data locality in memory hierarchies. 
It features two memory types, namely fast internal memory (IM or RAM) holding up to $M$ data items, and a slow disk of unbounded size.
The input and output of an algorithm are stored in EM while computation is only possible on values in IM.
An algorithm's performance is measured in the number of I/Os required.
Each I/O transfers a block of $B = \Omega(\sqrt{M})$ consecutive items between memory levels.
Reading or writing $n$ contiguous items is referred to as \emph{scanning} and requires $\scan(n) \defrel \Theta(n/B)$~I/Os.
Sorting $n$ consecutive items triggers $\sort(n) \defrel \Theta((n/B) \cdot \log_{M/B}(n/B))$~I/Os.
For all realistic values of $n$, $B$ and $M$, $\scan(n)<\sort(n)\ll n$.
Sorting complexity constitutes a lower bound for most intuitively non-trivial EM tasks~\cite{Meyer2003}.
EM queues use amortised $\O(1/B)$ I/Os per operation and require $\O(B)$ main memory~\cite{pagh2003basic}.
An external priority queue (PQ) requires $\O(\sort(n))$ I/Os to push and pop $n$ items~\cite{Arge1995}.

\subsection{\tfpHead: Time Forward Processing}\label{subsec:tfp}
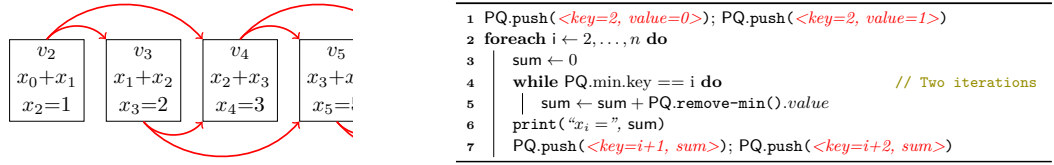
\begin{figure*}
  \vspace{-1em}  
  \begin{minipage}{0.35\textwidth}
    \scalebox{0.8}{%
\begin{tikzpicture}[
  event/.style = {draw, xshift=1em, anchor=west, align=center},
  depn/.style = {draw, ->, thick, bend left=50, red},
  deps/.style = {draw, ->, thick, bend right=50, red},
  depo/.style = {opacity=0.3, dashed}
  ]
  
  \begin{scope}
    \clip(0,-1.5) rectangle (6, 1.5);
    
    \node[event] at (0,0)     (x2) {$v_2$ \\ $x_0 {+} x_1$ \\ $x_2{=}1$};
    \node[event] at (x2.east) (x3) {$v_3$ \\ $x_1 {+} x_2$ \\ $x_3{=}2$};
    \node[event] at (x3.east) (x4) {$v_4$ \\ $x_2 {+} x_3$ \\ $x_4{=}3$};
    \node[event] at (x4.east) (x5) {$v_5$ \\ $x_3 {+} x_4$ \\ $x_5{=}5$};
    \node[event] at (x5.east) (x6) {$v_6$ \\ $x_4 {+} x_5$ \\ $x_6{=}8$};
    \node[event] at (x6.east) (x7) {$v_7$ \\ $x_5 {+} x_6$ \\ $x_7{=}13$};
    
    \path[depn] (x2.north) to (x3.north west);
    \path[depn] (x2.north) to (x4.north west);
    
    \path[deps] (x3.south) to (x4.south west);
    \path[deps] (x3.south) to (x5.south west);
    
    \path[depn] (x4.north) to (x5.north west);
    \path[depn] (x4.north) to (x6.north west);
    
    \path[deps] (x5.south) to (x6.south west);
    \path[deps] (x5.south) to (x7.south west);
  
  \end{scope}
\end{tikzpicture} %
}
  \end{minipage}
  \hfill
  \scalebox{0.7}{\begin{minipage}{0.8\textwidth}
      \footnotesize
\setlength{\interspacetitleruled}{-.4pt}%
\begin{myalgorithm}[H]
  \SetKwData{pq}{PQ}
  \SetKwData{i}{i}
  \SetKwData{sum}{sum}
  
  \SetKwFunction{print}{print}
  \SetKwFunction{ppush}{push}
  \SetKwFunction{ppop}{remove-min}
  
  \pq.\ppush{$\text{\textcolor{red}{<key=2, value=0>}}$}; \pq.\ppush{$\text{\textcolor{red}{<key=2, value=1>}}$}\;
  \ForEach{$\i \gets 2, \ldots, n$}{
    $\sum \gets 0$\;
    \While(\tcp*[f]{Two iterations}){\pq.min.key == i}{
      $\sum \gets \sum + \pq.\ppop{}.value$\;
    }
    
    \print{``$x_i = $'', \sum}\;
    \pq.\ppush{$\text{\textcolor{red}{<key=i+1, sum>}}$}; \pq.\ppush{$\text{\textcolor{red}{<key=i+2, sum>}}$}
} 
\end{myalgorithm}   \end{minipage}}
  
  \vspace{-0.5em}
  \caption{
    \textbf{Left:} Dependency graph of the Fibonacci sequence (ignoring base case).
    \textbf{Right:} Time Forward Processing to compute sequence.
  }
  \label{fig:tfp}
\end{figure*}

Time Forward Processing (\tfp) is a generic technique to manage data dependencies of external memory algorithms \cite{Maheshwari2003}.
Consider an algorithm computing values $x_1, \ldots, x_n$ in which the calculation of $x_i$ requires previously computed values.
One typically models these dependencies using a directed acyclic graph $G{=}(V,E)$.
Every node $v_i\in V$ corresponds to the computation of $x_i$ and an edge $(v_i, v_j)\in E$ indicates that the value $x_i$ is necessary to compute $x_j$.
For instance consider the Fibonacci sequence $x_0=0,\ x_1=1,\ x_i = x_{i-1} + x_{i-2}\ \forall i \ge 2$ in which each node $v_i$ with $i \ge 2$ depends on exactly its two predecessors (see Fig.~\ref{fig:tfp}).
Here, a linear scan for increasing $i$ suffices to solve the dependencies.

In general, an algorithm needs to traverse $G$ according to some topological order $\prec_T$ of nodes $V$ and also has to ensure that each $v_j$ can access values from all $v_i$ with $(v_i, v_j) \in E$.
The \tfp{} technique achieves this as follows:
as soon as $x_i$ has been calculated, messages of the form $\langle v_j, x_i \rangle$ are sent to all successors $(v_i, v_j) \in E$.
These messages are kept in a minimum priority queue sorting the items by their recipients according to $\prec_T$.
By construction, the algorithm only starts the computation $v_i$ once all predecessors $v_j \prec_T v_i$ are completed.
Since these predecessors already removed their messages from the PQ, items addressed to $v_i$ (if any) are currently the smallest elements in the data structure and can be dequeued.
Using a suited EM PQ~\cite{Arge1995}, TFP incurs $\O(\sort(k))$ I/Os, where $k$ is the number of messages sent.

\section{Randomisation schemes}
\label{section:randschemes}
Here, we summarise the randomisation schemes ESMC~\cite{Milo2003} and Curveball for simple undirected graphs~\cite{CarstensPhd}, and then discuss the notion of global trades. 
Since these algorithms iteratively modify random parts of a graph, they can be analysed as finite Markov chains.
It is well known that any finite, irreducible, aperiodic, and symmetric Markov chain converges to the uniform distribution on its state space (e.g. \cite{Levin2009}). 
Its \emph{mixing time} indicates the number of steps necessary to reach the stationary distribution.

\subsection{Edge-Switching}
\label{subsection:es}

\esmc{} is a state-of-the-art randomisation method with a wide range of applications, e.g.\ the generation of graphs \cite{Hamann2017a,Lancichinetti2008}, or the randomisation of biological datasets~\cite{Iorio2016}.
In each step, \esmc{} chooses two edges $e_1 = [u_1, v_1], e_2 = [u_2, v_2]$ and a direction $d \in \{0,1\}$ uniformly at random and rewires them into $\{u_1, u_2\}, \{v_1, v_2\}$ if $d{=}0$ and $\{u_1, v_2\}, \{v_1, u_2\}$ otherwise.
If a step yields a non-simple graph, it is skipped.
\esmc{}'s Markov chain is irreducible~\cite{Eggleton1981}, aperiodic and symmetric~\cite{Mihail2003} and hence converges to the uniform distribution on the space of simple graphs with fixed degree sequence. 
While analytic bounds on the mixing time~\cite{Greenhill2011, Greenhill2015} are impractical, usually a number of steps linear in the number of edges is used in practice~\cite{Ray2012}.

\subsection{Simple Undirected Curveball algorithm}
\label{subsection:cb}

Curveball is a novel randomisation method. 
In each step, two nodes trade their neighbourhoods, possibly yielding faster mixing times~\cite{CarstensPhd, Strona2014, Verhelst2008}.

\begin{mydef}[Simple Undirected Trade]
	\label{def:trade}
	Let $G = (V, E)$ be a simple graph, $A$ be its adjacency list representation, and $A_u$ be the set of neighbours of node $u$.
	A trade $t = (i, j, \sigma)$ from $A$ to adjacency list $B$ is defined by two nodes~$i$ and~$j$, and a permutation $\sigma\colon D_{ij} \to D_{ij}$ where $A_{i-j} \defrel A_i \setminus (A_j \cup \{j\})$ and $D_{ij} := A_{i-j} \cup A_{j-i}$.
	As shown in Fig.~\ref{fig:trade}, performing $t$ on $G$ results in 
        $B_i = (A_i {\setminus} A_{i-j}) \cup \{ x \mid x \in D_{ij}, \rank_\sigma(x) \le |A_{i-j}| \}$ and 
        $B_j = (A_j {\setminus} A_{j-i}) \cup \{ x \mid x \in D_{ij}, \rank_\sigma(x) >  |A_{i-j}| \}$.
	Since edges are undirected, symmetry has to be preserved: 
	for all $u \in A_i {\setminus} B_i$ the label $j$ in adjacency list $B_u$ is changed to $i$ and analogously for $A_j \setminus B_j$.
\end{mydef}
\begin{figure}[t]
  \vspace{-1em}
  \begin{center}
    \scalebox{0.9}{%
\begin{tikzpicture}[
    graphnode/.style={draw, circle, minimum width=1.5em, inner sep=0.2em, align=center, fill=white},
    tradenode/.style={graphnode, fill=black!10},
    tradable/.style={thick},
    traded/.style={tradable, red},
    label/.style={align=left, inner sep=0}
  ]
	\node[tradenode] (ai) {$i$};
	\node[tradenode, right of=ai] (aj) {$j$};
	\node[graphnode, above left of=ai] (a1) {$1$};
	\node[graphnode, above of=ai] (a2) {$2$};
	\node[graphnode, above of=aj] (a3) {$3$};
	\node[graphnode, above right of=aj] (a4) {$4$};
	\node[graphnode, right of=aj] (a5) {$5$};

  \coordinate (abetweenij) at ($(ai)!0.5!(aj)$);
  \node[graphnode, opacity=0, below right of=ai] (abelowi) {};
  \node[graphnode] at (abetweenij |- abelowi) (a6) {$6$};

	\node[tradenode] (bi) at (11.5,0) {$i$};
  \node[tradenode, right of=bi] (bj) {$j$};
  \node[graphnode, above left of=bi] (b1) {$1$};
  \node[graphnode, above of=bi] (b2) {$2$};
  \node[graphnode, left of=bi, opacity=0] (m2) {$2$};
  \node[graphnode, above of=bj] (b3) {$3$};
  \node[graphnode, above right of=bj] (b4) {$4$};
  \node[graphnode, right of=bj] (b5) {$5$};
  
  \coordinate (bbetweenij) at ($(bi)!0.5!(bj)$);
  \node[graphnode, opacity=0, below right of=bi] (bbelowi) {};
  \node[graphnode] at (bbetweenij |- bbelowi) (b6) {$6$};

	\draw (ai) -- (aj);
	\draw (ai) -- (a6);
	\draw (aj) -- (a6);
	\draw[tradable] (aj) -- (a5);
	
	\draw[tradable] (ai) -- (a1);
	\draw[tradable] (ai) -- (a2);
	\draw[tradable] (aj) -- (a3);
	\draw[tradable] (aj) -- (a4);

  \draw (bi) -- (bj);
  \draw (bi) -- (b6);
  \draw (bj) -- (b6);
  \draw[tradable] (bj) -- (b5);
  
  \draw[traded] (bi) -- (b3);
  \draw[traded] (bi) -- (b4);
  \draw[traded] (bj) -- (b1);
  \draw[traded] (bj) -- (b2);

	\draw[very thick,->] (a5) ++ (0.5,0) -- +(0.5,0);
  \draw[very thick,<-] (m2) ++ (-0.5,0) -- +(-0.5,0);

	\node[below of=a6, label] (Ai) {$A_i = \{1,2,6,j\}$ \\ $A_j = \{3,4,5,6,i\}$};
  \node[below of=b6, label] (Bi) {$B_i = \{3,4,6,j\}$ \\ $B_j = \{1,2,5,6,i\}$};
  \node[label, anchor=center] (sigA)  at ($(Ai)!0.5!(Bi)$) {$B_{i-j} = \{3,4\}$ \\ $B_{j-i} = \{1,2,5\}$};

  \node[text depth=0pt, text height=0.5em] (sig) at ($(a5)!0.5!(m2)$) {
    $\sigma(\underbrace{\fbox{1,2}}_{A_{i-j}}, \underbrace{\fbox{3,4,5}}_{A_{j-i}}) 
    \mapsto (\underbrace{\fbox{4,3}}_{B_{i-j}},\underbrace{\fbox{5,1,2}}_{B_{j-i}})$};
\end{tikzpicture} %
}
  \end{center}

  \vspace{-1em}
  \caption{
    The trade $(i,j,\sigma)$ between nodes~$i$ and~$j$ only considers edges to the disjoint neighbours $\{1, \ldots, 5\}$.
    For the reassigned disjoint neighbours we use the short-hand $B_{i-j} \defrel \{ x \mid x \in D_{ij}, \rank_\sigma(x) \le |A_{i-j}| \}$ and $B_{j-i} \defrel \{ x \mid x \in D_{ij}, \rank_\sigma(x) > |A_{i-j}| \}$.
    The triangle $(i, j, 6)$ is omitted as trading any of its edges would either introduce parallel edges, self loops, or result in no change at all.
    Then, the given $\sigma$ exchanges four edges.}
	\label{fig:trade}
\end{figure}
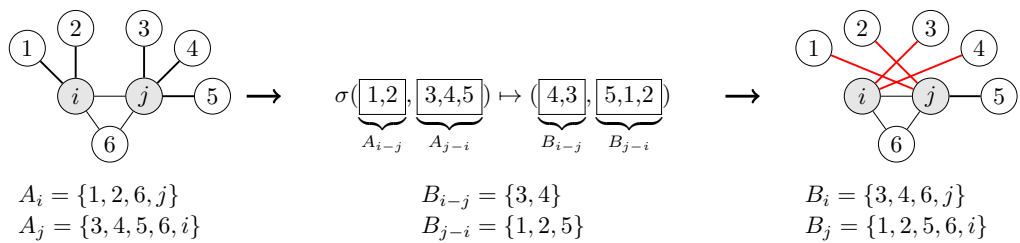

Simple Undirected Curveball randomises a graph by repeatedly selecting a node pair $\{i,j\}$ and permutation $\sigma$ on the disjoint neighbours uniformly at random.
Its Markov chain is irreducible, aperiodic and symmetric and hence converges to the uniform distribution \cite{Carstens2016}. 

\subsection{Undirected Global Trades}\label{subsec:gcb}
Trade sequences typically consist of pairs in which each constituent is drawn uniformly at random.
While it is a well-known fact\footnote{%
  For instance studied as the coupon collector problem.}
that $\Theta(n \log n)$ trades are required in expectation until each node is included at least once, there is no apparent reason why this should be beneficial;
in fact, experiments in \autoref{sec:experiments} suggest the contrary.

Carstens et al. propose the notion of \emph{global trades} for directed or bipartite graphs as a 2-partition of all nodes implicitly forming $n/2$ node pairs to be traded in a single step~\cite{Carstens2016}.
This concept fails for undirected graphs where in general the two directions $(u,v)$ and $(v, u)$ of an edge $\{u, v\}$ cannot be processed independently in a single step.
We hence extend global trades to undirected graphs by interpreting them as a sequence of $n / 2$ simple trades which together target each node exactly once (we assume $n$ to be even; if this is not the case we add an isolated node\footnote{%
  This is equivalent to randomly excluding a single node from a global trade}).
Dependencies are then resolved by the order of this sequence.
 
\begin{mydef}[Undirected Global Trade]
	\label{def:globaltrade}
  Let $G = (V, E)$ be a simple undirected graph and $\pi\colon V \to V$ be a permutation on the set of nodes.
  A \emph{global trade} $T = (t_1, \ldots, t_\ell)$ for $\ell = \lfloor n/2\rfloor$ is a sequence of trades $t_i = \{\pi(v_{2i {-} 1}), \pi(v_{2i}), \sigma_i\}$. 
  By applying $T$ to $G$ we mean that the trades $t_1, \ldots, t_\ell$ are applied successively starting with $G$. 
\end{mydef}

Theorem~\ref{thm:globaltrades} allows us to use global trades as a substitute for a sequence of single trades, as global trades preserve the stationary distribution of Curveball's Markov chain.
The proof  extends \cite{Carstens2016}, which shows convergence of global trades in bipartite or directed graphs, to undirected graphs and uses similar techniques.

\begin{mythm}
	\label{thm:globaltrades}
	Let $G = (V, E)$ be an arbitrary simple undirected graph, and let $\Omega_G$ be the set of all simple directed graphs that have the same degree sequence as $G$. The Curveball algorithm with global trades and started at $G$ converges to the uniform distribution on $\Omega_G$.
\end{mythm}
\begin{proof}
	In order to prove the claim, we have to show irreducibility and aperiodicity of the Markov chain as well as symmetry of the transition probabilities.

  For the first two properties it suffices to show that whenever there exists a single trade from state $A$ to $B$, there also exists a global trade from $A$ to $B$ (see \cite{Carstens2015} for a similar argument).\footnote{%
  	Since each global trade can be emulated by its $n/2$ decomposed single trades, the reverse is true for a hop of $n/2$ single trade steps.
  	Due to dependencies however the transition probabilities generally do not match, see $V = \{1, 2, 3, 4\}$ and $E = \{[1,2],[3,4]\}$ for a simple counterexample.
  }
  Observe that there is a non-zero probability that a single trade does not change the graph, e.g. by selecting $\sigma_i$ as the identity.
  Hence there is a non-zero probability that \ldots
  \begin{itemize}
    \item a global trade does not alter the graph at all. 
    This corresponds to a self-loop at each state of the Markov chain and hence guarantees aperiodicity.
    \item all but one single trade of a global trade do not alter the graph.
    In this case, a global trade degenerates to a single trade and the irreducibility shown in~\cite{Carstens2015} carries over.
  \end{itemize}

	It remains to show that the transition probabilities are symmetric. Let $\mathcal T_{AB}^g$ be the set of global trades that transform state $A$ to state $B$.
  Then the transition probability between $A$ and $B$ equals the sum of probabilities of selecting a trade sequence from $\mathcal T_{AB}^g$.
  That is $P_{AB} = \sum\nolimits_{T \in \mathcal T_{AB}^g}\mathbf P_A(T)$
	where $\mathbf P_A(T)$ denotes the probability of selecting global trade $T$ in state $A$. 

	The probability $\mathbf P_A(t)$ of selecting a single trade $t = (i, j, \sigma)$ from state $A$ to state $B$ equals the probability $\mathbf P_B(\tilde{t})$ of selecting the reverse trade $\tilde{t} = (i, j, \sigma^{-1})$ from state $B$ to $A$ \cite{Carstens2016}.
	We now define the reverse global trade of $T=(t_1, \ldots, t_\ell)$ as $\tilde{T} = (\tilde{t}_\ell, \ldots, \tilde{t}_1)$.
	It is straight-forward to check that this gives a bijection between the sets $\mathcal T^g_{AB}$ and $\mathcal T^g_{BA}$. 

	It remains to show that the middle equality holds in
	\begin{equation*}
	P_{AB} = \sum\nolimits_{T \in \mathcal T_{AB}^g}\mathbf P_A(T)\ \ \stackrel{!}{=}\ \  \sum\nolimits_{\tilde{T} \in \mathcal T_{BA}^g}\mathbf P_B(\tilde{T}) = P_{BA}.
	\end{equation*}

	Let $T=(t_1, \ldots, t_\ell)$ be a global trade from state~$A$ to state~$B$ as implied by $\pi$ and $A = A_1, \ldots, A_{\ell+1} = B$ be the intermediate states.
	We denote the reversal of $T$ and $\pi$ as $\tilde T$ and $\tilde \pi$ respectively and obtain 
	\[
		P_A(T) 
  = \mathbf P(\pi) \mathbf P_{A_1}(t_1) \ldots \mathbf P_{A_\ell}(t_\ell) 
	= \mathbf P(\tilde{\pi}) \mathbf P_{B}(\tilde{t}_\ell)  \ldots \mathbf P_{A_2}(\tilde{t}_1)    
  = P_B(\tilde{T})
	.\]
	Clearly $\mathbf P(\pi) = \mathbf P(\tilde{\pi})$ as we are picking permutations uniformly at random. The second equality follows from $\mathbf P_A(t) = \mathbf P_B(\tilde{t})$ for a single trade between $A$ and $B$. 
\end{proof} %
\section{Novel Curveball algorithms for undirected graphs}\label{sec:algo}
In this section we present the related algorithms \emcb{}, \imcb{}, \emgcb{} and \empgcb{}.
The algorithms receive a simple graph $G$ and a \emph{trade sequence} $T = \sequence{ \{u_i, v_i\} }{i=1}{\ell}$ as input and compute the result of carrying out the trade sequence $T$ (see section~\ref{subsection:cb}) in order.

\emcb{} and \imcb{} are sequential solutions suited to process arbitrary trade sequences~$T$.
For our analysis, we assume $T$'s constituents to be drawn uniformly at random (as expected in typical applications).
Both algorithms share a common design, but differ in the data structures used.
\emcb{} is an I/O-efficient algorithm while \imcb{} is optimised for small graphs allowing for unstructured accesses to main memory.
In contrast, \emgcb{} and  \empgcb{} process global trades only.
This restricted input model allows us to represent the trade sequence~$T$ implicitly by hash functions which further accelerates trading.

At core, all algorithms perform trades in a similar fashion: 
In order to carry out the $i$-th trade $\{u_i, v_i\}$, they retrieve the neighbourhoods $\Nei{u_i}$ and $\Nei{v_i}$, shuffle%
\footnote{%
  In contrast to Definition \ref{def:globaltrade}, we do not consider the permutation $\sigma$ of disjoint neighbours as part of the input, but let the algorithm choose one randomly for each trade.
  We consider this design decision plausible as the set of disjoint neighbours only emerges over the course of the execution.
}
them, and then update the graph.
Once the neighbourhoods are known, trading itself is straight-forward.
We compute the set of disjoint neighbours $D = (\Nei{u_i} \cup \Nei{v_i}) \setminus (\Nei{u_i} \cap \Nei{v_i})$ and then draw $|\Nei{u_i} \cap D|$ nodes from $D$ for $u_i$ uniformly at random while the remaining nodes go to $v_i$.
If $\Nei{u_i}$ and $\Nei{v_i}$ are sorted this requires only $\O(|\Nei{u_i}| + |\Nei{v_i}|)$ work and $\scan(|\Nei{u_i}| + |\Nei{v_i}|)$ I/Os (see also proof of Lemma~\ref{lem:imcb-complex} if the neighbourhoods fit into RAM).
Hence we focus on the harder task of obtaining and updating the adjacency information.

\subsection{\emcbHead: A sequential I/O-efficient Curveball algorithm}\label{subsec:algo-emcb}
\begin{myalgorithm}[t]
    \SetKwData{sortTtoV}{SorterTtoV}
    \SetKwData{sortDep}{SorterDepChain}
    \SetKwData{pqVtoV}{PqVtoV}
    \SetKwData{pqTtoT}{PqTtoT}

    \KwData{Trade sequence $T$, simple graph $G=(V, E)$ by edge list $E$} 
    
    \tcp{Preprocessing: Compute Dependencies}
    \ForEach{trade $t_i = (u, v) \in T$ for increasing $i$}{
        Send messages $\langle u, t_i \rangle$ and $\langle v, t_i \rangle$ to Sorter \sortTtoV\;
    }
    Sort \sortTtoV lexicographically \tcp*{All trades of a node are next to each other}
    \ForEach{node $u \in V$}{
        Receive $\Suc(u) = [t_1, \ldots, t_k]$ from $k$ messages addressed to $u$ in \sortTtoV\;
        Set $t_{k+1} \gets \infty$ \tcp*{$t_1 = \infty$ iff $u$ is never active} 
        Send $\langle t_i, u, t_{i+1} \rangle$ to \sortDep for $i \in [k]$\;
        
        \ForEach{directed edge $(u, v) \in E$}{
            \uIf{$u < v$}{
                Send message $\langle v, u, t_1 \rangle$ via \pqVtoV
            }\Else{
            Receive $t_1^v$ from unique message received via \pqVtoV\;
            $\arraycolsep=1.4pt\begin{array}{ll}
                \lIf{$t_1 \le t_1^v$}{} & \text{Send message $\langle t_1, u, v, t_1^v \rangle$ via \pqTtoT} \\ 
                \lElse{} &  \text{Send message $\langle t_1^v, v, u, t_1 \rangle$ via \pqTtoT}
            \end{array}$
    }  
}
}
Sort \sortDep\;

\tcp{Main phase -- Currently at least the first trade has all information it needs}
\ForEach{trade $t_i = (u, v) \in T$ for increasing $i$}{
    Receive successors $\tau(u)$ and $\tau(v)$ via \sortDep\;
    Receive neighbours $\AG(u)$, $\AG(v)$ and their successors $\tau(\cdot)$ from \pqTtoT\;
    
    Randomly reassign disjoint neighbours, yielding new neighbours $\AG'(u)$ and $\AG'(v)$.
    
    \ForEach{$(a, b) \in (\{u\} \times \AG'(u)) \cup (\{v\} \times \AG'(v))$}{
        $\arraycolsep=1.4pt\begin{array}{ll}
            \lIf{$\tau_a = \infty$ and $\tau_b = \infty$}{} & \text{Output final edge $\{a,b\}$} \\
            \lElseIf{$\tau_a \le \tau_b$}{} & \text{Send message $\langle \tau_a, a, b, \tau_b \rangle$ via \pqTtoT} \\
            \lElse{} & \text{Send message $\langle \tau_b, b, a, \tau_a \rangle$ via \pqTtoT}
        \end{array}$
    }
}

\caption{\emcb{}}
\label{algo:emcb}
\end{myalgorithm}

\emcb{} is an I/O-efficient Curveball algorithm to randomise undirected graphs as detailed in Alg.~\ref{algo:emcb}.
This basic algorithm already contains crucial design principles which we further explore with \imcb{}, \emgcb{} and \empgcb{} in sections~\ref{subsec:algo-imcb} and~\ref{subsec:algo-empgcb} respectively.

The algorithm encounters the following challenges.
After an undirected trade $\{u, v\}$ is carried out, it does not suffice to only update the neighbourhoods $\Nei{u}$ and $\Nei{v}$:
consider the case that edge $\{u, x\}$ changes into $\{v, x\}$.
Then this switch also has to be reflected in the neighbourhood of $\Nei{x}$.
Here, we call $u$ and $v$ \emph{active} nodes while $x$ is a \emph{passive} neighbour.

In the EM setting another challenge arises for graphs exceeding main memory;
it is prohibitively expensive to directly access the edge list since this unstructured pattern triggers $\Omega(1)$ I/Os for each edge processed with high probability.

\emcb{} approaches these issues by abandoning a classical static graph data structure containing two redundant copies of each edge.
Following the \tfp{} principle, we rather interpret all trades as a sequence of points over time that are able to receive messages.
Initially, we send each edge to the earliest trade one of its endpoints is active in.\footnote{%
  If an edge connects two nodes that are both actively traded we implicitly perform an arbitrary tie-break.}
This way, the first trade receives one message from each neighbour of the active nodes and hence can reconstruct $\Nei{u_1}$ and $\Nei{v_1}$.
After shuffling and reassigning the disjoint neighbours, \emcb{} sends each resulting edge to the trade which requires it next.
If no such trade exists, the edge can be finalised by committing it to the output.

The algorithm hence requires for each (actively or passively) traded node $u$, the index of the next trade in which $u$ is actively processed.
We call this the \emph{successor} of $u$ and define it to be $\infty$ if no such trade exists.
The dependency information is obtained in a preprocessing step; given $T = \sequence{ \{u_i, v_i\} }{i=1}{\ell}$, we first compute for each node $u$ the monotonically increasing index list $\Suc(u)$ of trades in which $u$ is actively processed, i.e. $\Suc(u) \defrel \big[\, i\, |\, u \in t_i \text{ for } i \in [\ell]\,\big] \circ [\infty]$.
\begin{myexp}
  Let $G = (V, E)$ be a simple graph with $V = \{v_1, v_2, v_3, v_4\}$ and trade sequence
  $T = [
  \textcolor{blue}  {t_1{:}}~\{v_1, v_2\},
  \textcolor{red}   {t_2{:}}~\{v_3, v_4\},
  \textcolor{olive} {t_3{:}}~\{v_1, v_3\},
  \textcolor{orange}{t_4{:}}~\{v_2, v_4\},
  {t_5{:}}~\{v_1, v_4\}
  ]$.
  Then, the successors $\Suc$ follow as
  $\Suc(v_1) = [\textcolor{blue}{1}, \textcolor{olive}{3}, 5, \textcolor{gray}{\infty}]$,
  $\Suc(v_2) = [\textcolor{blue}{1}, \textcolor{orange}{4}, \textcolor{gray}{\infty}]$,
  $\Suc(v_3) = [\textcolor{red}{2}, \textcolor{olive}{3}, \textcolor{gray}{\infty}]$,
  $\Suc(v_4) = [\textcolor{red}{2}, \textcolor{orange}{4}, 5, \textcolor{gray}{\infty}]$.
\end{myexp}

\noindent This information is then spread via two channels:
\begin{itemize}
  \item After preprocessing, \emcb{} scans $\Suc$ and $T$ conjointly and sends $\langle t_i, u_i, t^u_i \rangle$ and $\langle t_i, v_i, t^v_i \rangle$ to each trade $t_i$.
  The messages carry the successors $t^u_i$ and $t^v_i$ of the trade's active nodes.
  
  \item When sending an edge as described before, we augment it with the successor of the passive node.
  Initially, this information is obtained by scanning the edge list $E$ and $\Suc$ conjointly.
  Later, it can be inductively computed since each trade receives the successors of all nodes involved.
\end{itemize}

\begin{lemma}\label{lem:emcb-complex}
 	For an arbitrary trade sequence $T$ of length $\ell$, \emcb{} has a worst-case I/O complexity of
  $\O[\sort(\ell) + \sort(n) + \scan(m) + \ell d_{\max} / B \log_{M/B} (m / B)]$. 
  For $r$ global trades, the worst case I/O complexity is $\O(r [\sort(n) + \sort(m)])$.
\end{lemma}  

\begin{proof}
  Refer to Appendix \ref{sec:app-emcb} for the proof.
\end{proof}

\subsection{\imcbHead: An internal memory version of \emcbHead}\label{subsec:algo-imcb}
While \emcb{} is well-suited if memory access is a bottleneck, we also consider the modified version \imcb{}.
As shown in \autoref{sec:experiments}, \imcb{} is typically faster for small graph instances.
\imcb{} uses the same algorithmic ideas as \emcb{} but replaces its priority queues and sorters\footnote{%
  The term \emph{sorter} refers to a container with two modes of operation:
  in the first phase, items are pushed into the write-only sorter in an arbitrary order by some algorithm.
  After an explicit switch, the filled data structure becomes read-only and the elements are provided as a lexicographically non-decreasing stream which can be rewound at any time.
  While a sorter is functionally equivalent to filling, sorting and reading back an EM vector, the restricted access model reduces constant factors in the implementation's runtime and I/O-complexity~\cite{Beckmann2009}.
} by unstructured I/O into main memory (see Alg.\ref{algo:imcb} (Appendix) for details):
\begin{itemize}
  \item
    Instead of sending neighbourhood information in a TFP-fashion, we now rely on a classical adjacency vector data structure $\AG$ (an array of arrays).
    Similarly to \emcb{}, we only keep one directed representation of an undirected edge.
    As an invariant, an edge is always placed in the neighbourhood of the incident node traded before the other.
    To speed-up these insertions, \imcb{} maintains unordered neighbourhood buffers.
    
  \item
    \imcb{} does not forward successor information, but rather stores $\Suc$ in a contiguous block of memory.
    The algorithm additionally maintains the vector $\Suc_\text{idx}[1 \ldots n]$ where the $i$-th entry points to the current successor of node $v_i$.
    Once this trade is reached, the pointer is incremented giving the next successor.
  
\end{itemize}

\begin{mylem}\label{lem:imcb-complex}
	For a random trade sequence $T$ of length $\ell$, \imcb{} has an expected running time of $\O(n + \ell + m + \ell m/n)$.
	In the case of $r$ many global trades (each consisting of $n/2$ normal trades) the running time is given by
	$\O(n + r m)$.
\end{mylem}
\begin{proof}
  
  Refer to Appendix \ref{sec:app-imcb} for the proof.
\end{proof}

\subsection{\emgcbHead{}: An I/O-efficient Global Curveball algorithm}\label{subsec:algo-emgcb}
\begin{figure}
    \begin{center}  
        \vspace{-1em}
        \scalebox{0.75}{%
\begin{tikzpicture}[
  cell/.style={draw, outer sep=0,
    minimum width=2.5em, minimum height=1em, node distance=0pt, font=\Large},
  pindex/.style={gray, node distance=0pt, font=\footnotesize},
  edgetext/.style={node distance=0pt,font=\Large,text depth=0pt, text height=0.5ex}  
]
  \newcommand{\drawperm}[3]{
    \node[] (p#1-0) at (#2) {};
    \path[draw, thick] ($(p#1-0.east) - (0, 1em)$) to ($(p#1-0.east) + (0, 1em)$);
    \foreach \x/\p [count=\i] in {#3} {
       \node[cell, right=of p#1-\intcalcDec{\i}] (p#1-\i) {$v_\p$};
       \node[pindex, above=of p#1-\i] {$\pi_#1(\i)$};
        
       \ifthenelse{\equal{\intcalcMod{\i}{2}}{0}}{
         \path[draw, thick] ($(p#1-\i.east) - (0, 1em)$) 
           to ($(p#1-\i.east) + (0, 1em)$);
       }{}
    }
  }

  \drawperm{1}{0,0}{3,1,2,5,4,6}
  \drawperm{2}{24em,0}{6,3,5,1,2,4}
  
  \draw [decorate,decoration={brace,amplitude=10pt,raise=1em}]
   (p1-3.north west) to node [black,midway,above,yshift=2em,align=center,font=\small ] {current trade} (p1-4.north east);

  \node[edgetext] (et-comma) at (17em, -5em) {$,$};
  \node[edgetext, left =of et-comma] (et-vl) {$v_1$};
  \node[edgetext, right=of et-comma] (et-vr) {$v_2$};

  \node[edgetext, left =of et-vl] {new edge produced: \Large$\{$};
  \node[edgetext, right=of et-vr] {$\}$};

  \path[draw, ->, dashed, out=145, in=290, gray] (et-vl.north) to (p1-2.south);
  \path[draw, ->, out=30, in=220] (et-vl.north) to node [above, xshift=-1.5em, yshift=-0.1em] {
    \small $\langle \text{\footnotesize round: } 2, \text{\footnotesize slot: } 4, \text{\footnotesize neighbour: } v_2 \rangle$} (p2-4.south);

  \path[draw, ->, dashed, out=125, in=290, gray] (et-vr.north) to (p1-3.south);
  \path[draw, ->, dashed, out=20, in=230, gray] (et-vr.north) to (p2-5.south);
\end{tikzpicture} %
}
        \vspace{-1.5em}
    \end{center}
    \caption{%
        During the trade $j{=}1, i_1{=}3, i_2{=}4$ the edge $\{v_1, v_2\}$ is produced; the arrows indicate positions considered as successors.
        Since $v_1$ and $v_2$ are already processed in round $j{=}1$, $\pi_2$ is used to compute the successor.
        Then, the message is sent to $v_1$ in round~2 as $v_1$ is processed before $v_2$.
    }
    \label{fig:permutation_trading}
\end{figure}
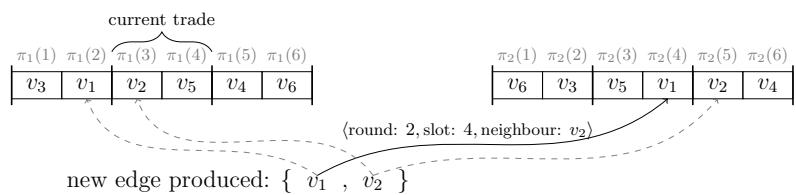

\emgcb{} builds on \emcb{} and exploits the regular structure of global trades to simplify and accelerate the dependency tracking.
As discussed in section~\ref{subsec:gcb}, a global trade can be encoded as a permutation $\pi\colon [n] \rightarrow [n]$ by interpreting adjacent ranks as trade pairs, i.e. $T_\pi = \sequence{\{v_{\pi(2i-1)}, v_{\pi(2i)}\}}{i=1}{n/2}$.
In this setting, a sequence of global trades is given by $r$ permutations $\sequence{\pi_j}{j=1}{r}$.
The model simplifies dependencies as it is not necessary to explicitly gather $\Suc$ and communicate successors.

As illustrated in Fig.~\ref{fig:permutation_trading}, we also change the addressing scheme of messages.
While \emcb{} sends messages to specific nodes in specific trades, \emgcb{} exploits that each node $v_i$ is actively traded only once in each round~$j$ and hence can be addressed by its position~$\pi_j(i)$.
Successors can then be computed in an ad hoc fashion;
let a trade of adjacent positions $i_1 < i_2$ of the $j$-th global trade produce (amongst others) the edge $\{v_x, v_y\}$.
The successor of $v_x$ (and analogously the one of $v_y$) is
$\Suc_{j,i_2}[v_x] = (j, \pi_{j}(x))$ if $v_x$ is processed later in round~$j$ (i.e. $\pi_j(x) / 2 > i_2$) and otherwise $\Suc_{j,i_2}[v_x] = (j{+}1, \pi_{j+1}(x))$.
Here we imply an untraded additional function $\pi_{r+1}(x) = x$ which avoids corner cases and generates an ordered edge list as a result of the $r$-th global trade.

To reduce the computational cost of the successor computation, \emgcb{} supports fast injective functions $f\colon X \rightarrow Y$ where $[n] \subseteq X$ and $[n] \subseteq Y$.
In contrast to the original permutations, their relevant image $\{\, f(x) \ |\  x\in [n]\,\}$ may contain gaps which are simply skipped by \emgcb{}.
This requires minor changes in the addressing scheme (see Appendix~\ref{sec:app-emgcb}).

\goodbreak

In practice, we use functions from the family of linear congruential maps $H_p$ where $p$ is the smallest prime number $p \ge n$:
\begin{align}
	H_p &\defrel& \left\{\, h_{a,b}\ \middle|\ 1 \le a < p \text{ and } 0 \le b < p\, \right\} \\
	h_{a,b}(x) &\equiv& (ax + b) \mod p,
\end{align}

As detailled in Appendix \ref{sec:lincongmap} random choices from $H_p$ are well suited for \emgcb{} since they are 2-universal\footnote{i.e. given one node in a single trade, the other is uniformly chosen among the remaining nodes.} and contain only $\O(\log(n))$ gaps.
They are also bijections with an easily computable inverse $h^{-1}_{a,b}$ that allows \emgcb{} to determine the active node  $h^{-1}_{a,b}(i)$ traded at position $i$; this operation is only performed once for each traded position. 
\emgcb{} also supports non-invertible functions.
This can be implemented with messages $\langle h(i), i \rangle$ that are generated for $1 \le i \le n$ and delivered using TFP.

\subsection{\empgcbHead: An I/O-efficient parallel Global Curveball algorithm}\label{subsec:algo-empgcb}
\empgcb{} adds parallelism to \emgcb{} by concurrently executing multiple sequential trades. 
As in Fig.~\ref{fig:microchunks}, we split a global trade into \emph{microchunks} each containing a similar number of node pairs and then execute a \emph{batch} of $p$ such subdivisions in parallel.
The batch's size is a compromise between intra-batch dependencies (messages are awaited from another processor) and overhead caused by synchronising threads at the batch's end (see Appendix \ref{sec:app-empgcb}).

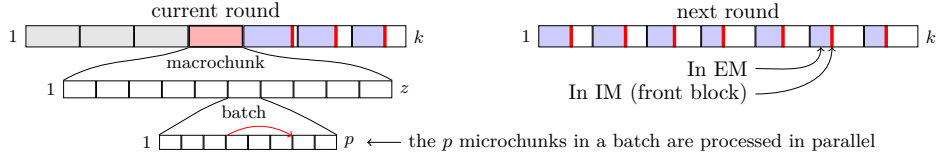
\begin{figure}[t]
  \begin{center}
    \vspace{-1.5em}
    \scalebox{0.8}{%
\def\macroWidth{2.5em}
\begin{tikzpicture}[
  every node/.style={node distance=},
  macro/.style={draw, minimum width=\macroWidth, minimum height=1em},
  batch/.style={draw, minimum width=1.5em, minimum height=0.8em},
  micro/.style={draw, minimum width=1em, minimum height=0.5em},
  pindex/.style={gray, node distance=0pt, font=\footnotesize},
  edgetext/.style={node distance=0pt,font=\Large,text depth=0pt, text height=0.5ex},
  onif/.code args={<#1>#2}{\ifthenelse{#1}{\pgfkeysalso{#2}}{}}
]
  \newcommand{\drawmacro}[3]{
    \node[] (macro#1-0) at (#2) {};
    \foreach \state/\f [count=\i] in {#3} {
      \node[opacity=0, onif={<\equal{\state}{p}>, opacity=1},
        minimum width=\f*\macroWidth, minimum height=1em,
        anchor=west, fill=blue!20]
        (macro#1-\i-fill) at (macro#1-\intcalcDec{\i}.east) {};
      \path[onif={<\equal{\state}{p}>, draw},
            red, ultra thick] (macro#1-\i-fill.north east) to (macro#1-\i-fill.south east);

      \node[macro,
            onif={<\equal{\state}{d}> fill=black!10},
            onif={<\equal{\state}{a}> fill=red!30, thick},
            onif={<\equal{\state}{p}> },
            right=of macro#1-\intcalcDec{\i}] (macro#1-\i) {};

      \xdef\macroNum{\i}
    }

    \node[left=of macro#1-1] {\footnotesize 1};
    \node[right=of macro#1-\macroNum] {\footnotesize $k$};
  }
  
  \drawmacro{1}{0,0   }{d/1,   d/1,   d/1,   a/1,   p/0.9, p/0.7, p/0.6}
  \drawmacro{2}{24em,0}{p/0.6, p/0.6, p/0.5, p/0.4, p/0.5, p/0.4, p/0.4}
  
  \node[below=of macro2-4.south east, anchor=east, align=right, yshift=-1em] (label-em) {In EM};
  \node[below=of macro2-4.south east, anchor=east, align=right, yshift=-2.2em] (label-im)  {In IM (front block)};
  
  \path[draw, ->, out=0, in=270] (label-em) to (macro2-6-fill.south);
  \path[draw, ->, out=0, in=270] (label-im) to (macro2-6-fill.south east);

  \node[anchor=south, above=of macro1-4] {current round};
  \node[anchor=south, above=of macro2-4] {next round};
  
  \node[yshift=-2.5em, xshift=-7.5em] (batch-0) at (macro1-4) {};
  \foreach \i in {1,...,10} {
    \node[batch, right=of batch-\intcalcDec{\i}] (batch-\i) {};
  }

  \node[left=of batch-1] {\footnotesize 1};
  \node[right=of batch-10] {\footnotesize $z$};

  \node[anchor=south, below=of macro1-4] {\footnotesize macrochunk};
  \path[draw, out=30,  in=200] (batch-1.north west) to (macro1-4.south west);
  \path[draw, out=150, in=340] (batch-10.north east) to (macro1-4.south east);

  \node[yshift=-2.5em, xshift=-4.5em, micro, opacity=0] (micro-0) at (batch-6) {};
  \foreach \i in {1,...,8} {
    \node[micro, right=of micro-\intcalcDec{\i}] (micro-\i) {};
  }

  \node[left=of micro-1] {\footnotesize 1};
  \node[right=of micro-8] {\footnotesize $p$};

  \node[anchor=south, below=of batch-6] {\footnotesize batch};
  \path[draw, out=30,  in=200] (micro-1.north west) to (batch-6.south west);
  \path[draw, out=150, in=340] (micro-8.north east) to (batch-6.south east);

  \path[draw, bend left, ->, red] (micro-3) to (micro-7);

  \node[right=of micro-8, xshift=3em, align=left, font=\small] (micro-label) {the $p$ microchunks in a batch are processed in parallel};
  \path[draw, ->, shorten >=0.5cm] (micro-label) to (micro-8);
\end{tikzpicture} %
}
    \vspace{-1.5em}
  \end{center}
  
  \caption{%
    \empgcb{} splits each global trade into $k$ \emph{macrochunks} and maintains an external memory queue for each.
    Before processing a macrochunk, the buffer is loaded into IM and sorted, and further subdivided into $z$ batches each consisting of $p$ microchunks.
    A type (ii) message is visualised by the red intra-batch arrow.
  }
  \label{fig:microchunks}
\end{figure}

\empgcb{} processes each microchunk similarly as in \emcb{} but differentiates between messages that are sent (i) within a microchunk, (ii) between microchunks of the same batch (iii) and microchunks processed later.
Each class is transported using an optimised data structure (see below) and only type (ii) messages introduce dependencies between parallel executions and are resolved as follows:
each processor retrieves the messages that are sent to its next trade and checks whether all information required is available by comparing the number of messages to the active nodes' degrees.
If data is missing the trade is skipped and later executed by the processor that adds the last missing neighbour.

For graphs with $m = \O(M^2/B)$ edges\footnote{%
  Even with as little as 1~GiB of internal memory, several billion edges are supported.
}, we optimise the communication structure for type~(iii) messages.
Observe that \empgcb{} sends messages only to the current and the subsequent round.
We partition a round into $k$ \emph{macrochunks} each consisting of $\Theta(n / k)$ contiguous trades.
An external memory queue is used for each macrochunk to buffer messages sent to it; in total, this requires $\Theta(k B)$ internal memory.
Before processing a macrochunk, all its messages are loaded into IM, subsequently sorted and arranged such that missing messages can be directly placed to the position they are required in.
This can also be overlapped with the processing of the previous macrochunk. 
As thoroughly discussed in Appendix \ref{sec:app-empgcb}, the number $k$ of macrochunks should be as small as possible to reduce overheads, but sufficiently large such that all messages of a macrochunk fit into main memory (see Appendix~\ref{section:analysis}).

\begin{mythm}
  \empgcb{} requires $\O(r \cdot [\sort(n) + \sort(m)])$ I/Os to perform $r$ global trades.
\end{mythm}

\begin{proof}
  Observe that we can analyse each of the $r$ rounds individually.
  A constant amount of auxiliary data is needed per node to provision gaps for missing data, to detect whether a trade can be executed and (if required) to invert the permutation. 
  This accounts for $\Theta(n)$ messages requiring $\sort(n)$ I/Os to be delivered.
  Using an ordinary PQ, the analysis of $\emcb$ (see Lemma~\ref{lem:emcb-complex}) carries over, requiring $\sort(m)$ I/Os for a global trade.
\end{proof}

\section{Experimental Evaluation}\label{sec:experiments}
In this section we evaluate the quality of the proposed algorithms and analyse the runtime of our C++ implementations.\footnote{Code used for the presented benchmarks can be found at our fork \url{https://github.com/hthetran/networkit} (\imcb{} and \emcb{}) and \url{https://github.com/massive-graphs/extmem-lfr} (\empgcb{}).}
\emcb{}, \imcb{}, \emgcb{} are designed as modules of NetworKit~\cite{Staudt2016}; due to their superior performance, only the latter two were added to the library and are available since release~4.6.
\empgcb{}'s implementation is developed separately and facilitates external memory data structures and algorithms of STXXL~\cite{Dementiev2008}.

Intuitively, graphs with skewed degree distributions are hard instances for Curveball since it shuffles and reassigns the disjoint neighbours of two trading nodes.
Hence, limited progress is achieved if a high-degree node trades with a low-degree node.
Since our experiments support this hypothesis, we focus on graphs with powerlaw degree distributions as difficult but highly relevant graph instances.
Our experiments use two parameter sets:
\begin{itemize}
	\item \textsl{(lin)} $-$ The maximal possible degree scales linearly as a function of the number $n$ of nodes.
  The degree distribution $\pld{a}{b}{\gamma}$ is chosen as $a = 10$, $b = n/20$ and $\gamma = 2$.
	\item \textsl{(const)} $-$ The extremal degrees are kept constant.
  In this case the parameters are chosen as $a = 50$, $b = 10000$ and $\gamma = 2$.
\end{itemize}
We select these configurations to be comparable with \cite{Hamann2017a} where both parameter sets are used to evaluate \emes{}.
The first setting \textsl{(lin)} considers the increasing average degree of real-world networks as they grow.
The second setting \textsl{(const)} approximates the degree distribution of the Facebook network in May~2011 (refer to \cite{Hamann2017b} for details).
Runtimes are measured on the following off-the-shelf machine:
Intel Xeon E5-2630~v3 (8 cores at 2.40GHz), 64GB RAM, 2$\times$ Samsung 850~PRO SATA~SSD~(1~TB), Ubuntu Linux 16.04, GCC~7.2.

\subsection{Mixing of Edge-Switching, Curveball and Global Curveball}
\label{subsection:empgcbvsemes}

\begin{figure}[t]
	\centering
	\includegraphics[scale=0.68]{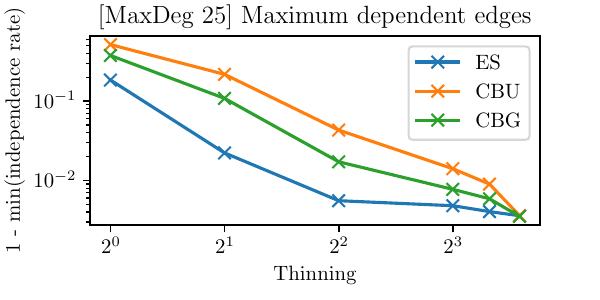}\hfill
	\includegraphics[scale=0.68]{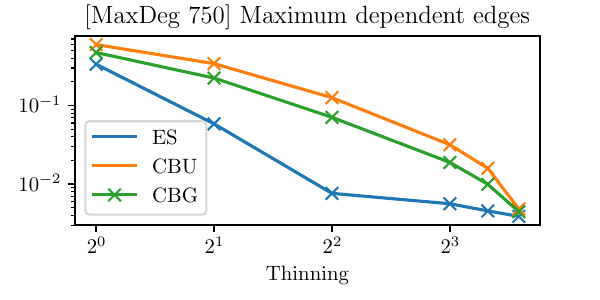}
	\vspace{-0.5em}
	\caption{
		Fraction of edges still correlated as a function of the thinning parameter $k$ for graphs with $n = 2 {\cdot} 10^3$ nodes and degree distribution $\pld{a}{b}{\gamma}$ with $\gamma = 2$, $a = 5$, and $b \in \{25, 750\}$.
		The (not thinned) long Markov chains of edge switching (ES), Curveball with uniform trades (CBU) and Curveball with global trades (CBG) contain 6000 super steps each.
	}
	\label{fig:independentedges}
\end{figure}

We are not aware of any practical theoretical bounds on the mixing time of Markov chains of Curveball, Global Curveball or edge switching (see \autoref{section:randschemes}).
Hence, we quantitatively study the progress made by Curveball trades compared to edge switching and approximate the mixing time of the underlying Markov chains by a method  developed in \cite{Ray2015}.
This criterion is a more sensitive proxy to the mixing time than previously used alternatives, such as the local clustering coefficient, triangle count and degree assortativity~\cite{Hamann2017b}.

Intuitively, one determines the number of Markov chain steps required until the correlation to the initial state decays.
Starting from an initial graph $G_0$, the Markov chain is executed for a large number of steps, yielding a sequence $(G_t)_{t\ge0}$ of graphs evolving over time.
For each occurring edge $e$, we compute a boolean vector $(Z_{e,t})_{t\ge0}$ where a~$1$ at position~$t$ indicates that~$e$ exists in graph $G_t$.
We then derive the \textsl{$k$-thinned} series $(Z_{e,t}^k)_{t\ge0}$ only containing every $k$-th entry of the original vector $(Z_{e,t})_{t\ge0}$ and use $k$ as a proxy for the mixing time.

To determine if $k$ Markov chain steps suffice for edge~$e$ to lose the correlation to the initial graph, the empirical transition probabilities of the $k$-thinned series $(Z_{e,t}^k)_{t\ge0}$ are fitted to both an independent and a Markov model respectively.
If the independent model is a better fit, we deem edge~$e$ to be independent.

The results presented here consider only small graphs due to the high computational cost involved.
However, additional experiments suggest that the results hold for graphs at least one order of magnitude larger which is expected as powerlaw distributions are scale-free.

We compare a sequence of uniform (single) trades, global trades and edge switching and visually align the results of these schemes by defining a \emph{super step}.
Depending on the algorithm a super step corresponds to either a single global trade, $n/2$ uniform trades or $m$ edge-swaps.
Comparing $n/2$ uniform trades with a global trade seems sensible since a global trade consists of exactly $n/2$ single trades, furthermore randomising with $n/2$ single trades considers the state of $2m$ edges which is also true for $m$ edge-swaps.
The alignment accounts for the fact that a single Curveball Markov chain step may execute multiple neighbour switches, thus easily outperforming \esmc{} in a step-by-step comparison.

Fig.~\ref{fig:independentedges} contains a selection of results obtained for small powerlaw graph instances using this method (see Appendix \ref{sec:experiment-compare-global-trades-to-trades} for the complete dataset).
Progress is measured by the fraction of edges that are still classified as correlated, i.e. the faster a method approaches zero the better the randomisation.
We omit an in-depth discussion of uniform trades and rather focus on global trades which consistently outperform the former (cf. section~\ref{subsection:cb}).

In all settings \esmc{} shows the fastest decay.
The gap towards global trades growths temporarily as the maximal degree is increased which is consistent with our initial claim that skewed degree distributions are challenging for Curveball.
The effect is however limited and in all cases performing $4$~global trades for each edge switching super step gives better results.
This is a pessimistic interpretation since typically $10m$ to $100m$ edge switches are used to randomise graphs in practice;
in this domain global trades perform similarly well and $20$~global trades consistently give at least the quality of $10m$ edge switches.

\subsection{Runtime performance benchmarks}
We measure the runtime of the algorithms proposed in section~\ref{sec:algo} and compare them to two state-of-the-art edge switching schemes (using the authors' C++ implementations):
\begin{itemize}
  \item 
    \vles{} is a sequential IM algorithm with a hashing-based data structure optimised for efficient neighbourhood queries and updates~\cite{DBLP:journals/corr/abs-cs-0502085}.
    To achieve comparability, we removed connectivity tests, fixed memory management issues, and adopted the number of swaps.
  \item
  \emes{} is an EM edge switching algorithm and part of \emlfr{}'s toolchain~\cite{Hamann2017a}.
\end{itemize}  

We carry out experiments using the \textsl{(const)} and \textsl{(lin)} parameter sets, and limit the problem sizes for internal memory algorithms to avoid exhaustion of the main memory.
For each data point we carry out 10~super steps (i.e. 10~global trades or $10m$ edge swaps) on a graph generated with Havel-Hakimi from a random powerlaw degree distribution.

Figure~\ref{fig:runtimes} presents the walltime per edge and super step including pre-computation\footnote{For \vles{} we report only the swapping process and the generation of the internal data structures.} required by the algorithms but excluding the initial graph generation process. The plots include (mostly small) errorbars corresponding to the unbiased estimation of the standard deviation of $S$ repetitions per data point (with different random seeds).
\begin{figure}[t]
    \centering
    \includegraphics[scale=0.5]{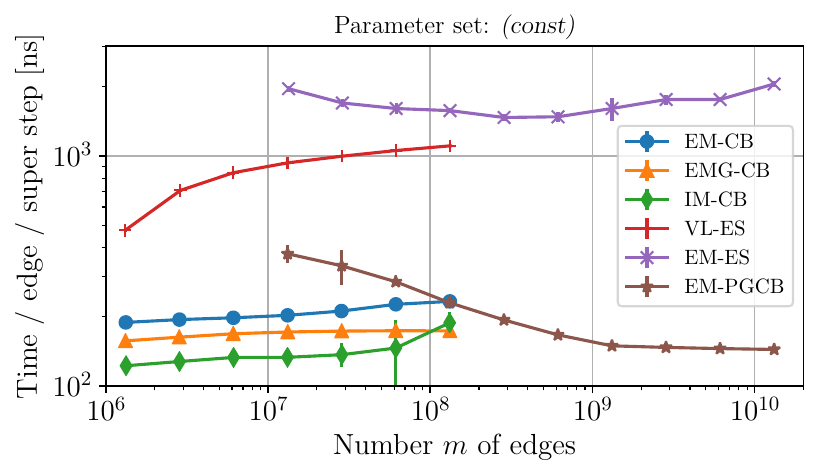}\hfill
    \includegraphics[scale=0.5]{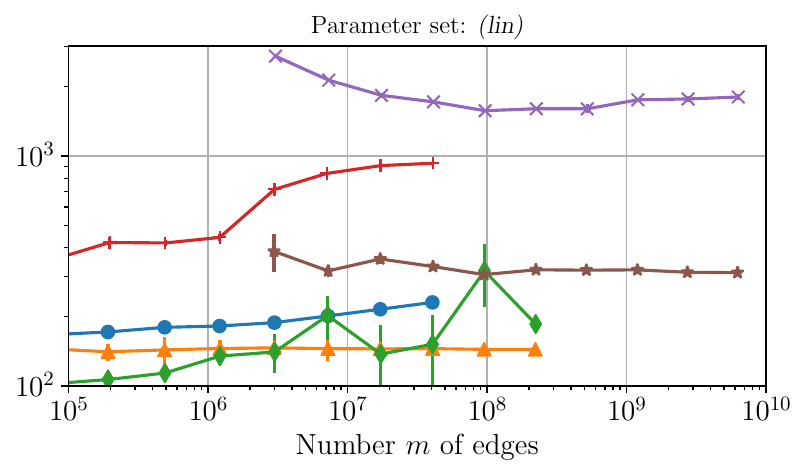}
    \caption{
        Runtime per edge and super step (global trade or $m$ edge swaps) of the proposed algorithms \imcb{}, \emcb{} and \empgcb{} compared to state-of-the-art IM edge switching \vles{} and EM edge switching \emes{}. 
        Each data point is the median of $S \ge 5$ runs over 10~super steps each.    
        The left plot contains the (const)-parameter set, the right one (lin).
        Observe that the super steps of different algorithms advance the randomisation process at different speeds (see discussion).
    }
    \label{fig:runtimes}
\end{figure}

The number $k$ of macrochunks does not significantly affect \empgcb{}'s performance for small graphs due to comparably high synchronisation cost.
In contrast, adjusting $k$ for larger graphs can noticeably increase the performance of \empgcb{}.
We thus experimentally determined the value $k = 32$ for both \textsl{(const)} and \textsl{(lin)} with $n = 10^7$ nodes and use that value for all other instances.

All Curveball algorithms outperform their direct competitors significantly --- even if we pessimistically executed two global trades for each edge switching super step (see section~\ref{subsection:empgcbvsemes}).
For large instances of \textsl{(const)} \empgcb{} carries out one super step $14.3$ times faster than \emes{} and $5.8$ times faster for \textsl{(lin)}. \empgcb{} also shows a superior scaling behaviour with an increasing speed-up for larger graphs.
Similarly, \imcb{} processes super steps up to $6.3$ times faster than \vles{} on \textsl{(const)} and $5.1$ times on \textsl{(lin)}.

On our test machine, the implementation of \imcb{} outperforms \emcb{} in the internal memory regime; \emgcb{} is faster for large graphs.
As indicated in Fig. \ref{fig:runtimes-thinkpen} (Appendix \ref{subsec:thinkpen-results}), this changes qualitatively for machines with slower main memory and smaller cache;
on such systems the unstructured I/O of \imcb{} and \vles{} is more significant rendering \emcb{} and \emgcb{} the better choice with a speed-up factor exceeding $8$ compared to \vles.
\section{Conclusion and outlook}
We applied \emph{global} Curveball trades to undirected graphs simplifying the algorithmic treatment of dependencies and showed that the underlying Markov chain converges to a uniform distribution.
Experimental results show that global trades yield an improved quality compared to a sequence of uniform trades of the same size.

We presented \imcb{} and \emcb{}, the first efficient algorithms for Simple Undirected Curveball algorithms; they are optimised for internal and external memory respectively.
Our I/O-efficient parallel algorithm \empgcb{} exploits the properties of global trades and executes a super step $14.3$ times faster than the state-of-the-art edge switching algorithm \emes{}; 
for \imcb{} we demonstrate speed-ups of up to $6.3$ (in a conservative comparison the speed-ups should be halved to account for the differences in mixing times of the underlying Markov chains).
The implementations of all three algorithms are freely available and are in the process of being incorporated into \emlfr{} and considered for NetworKit. 
\section*{Acknowledgments}
We thank the anonymous reviewers for their many insightful comments and suggestions.

\clearpage
\appendix
\section{Appendix: \emcbHead{}}\label{sec:app-emcb}

\begin{proof}[Proof of Lemma~\ref{lem:emcb-complex}]
  As in Alg.~\ref{algo:emcb}, \emcb{} scans $T$ and $E$ during preprocessing thereby triggering $\O(\scan(\ell) + \scan(m))$ I/Os.
  It also involves sorters \texttt{SorterTtoV} and \texttt{SorterDepChain} as well as priority queues \texttt{PqVtoV} and \texttt{PqTtoT} transporting $\O(\ell)$, $\O(\ell)$, $\O(n)$ and $\O(n)$ messages respectively. Hence preprocessing incurs $\O(\sort(\ell) + \sort(n) + \scan(m))$ I/Os.
  
  During the $i$-th trade $\O(\deg(u_i) + \deg(v_i))$ messages are retrieved shuffled and redistributed causing $\O[\sort(\deg(u_i) + \deg(v_i))]$ I/Os.
  The bound can be improved to  $\O((\deg(u_i) + \deg(v_i))/B\log_{M/B}(m/B))$ by observing that $\O(m)$ items are stored in the PQ at any time.
  For a worst-case analysis we set $\deg(u_i) = \deg(v_i) = d_{\max}$ yielding the first claim.
  
  In case of $r$ global trades, preprocessing can be performed in $r$ chunks of $n/2$ trades each.
  By arguments similar to the previous analysis, this yields an I/O complexity of $\O(r \sort(n) + r \scan(m))$.
  For the main phase, the above analysis tightens to $\O(r\sort(m))$ using the fact that a single global trade targets each edge at most twice.
\end{proof}

\section{Appendix: \imcbHead{}}
\label{sec:app-imcb}

\begin{myalgorithm}[H]
  \SetKwData{sucidx}{$\Suc_\text{idx}$}
  \SetKwData{sucbeg}{$\Suc_\text{begin}$}
  
  \SetKwData{adjidx}{$\mathcal A_\text{idx}$}
  \SetKwData{adjbeg}{$\mathcal A_\text{begin}$}
 
	\KwData{Trade sequence $T$, simple graph $G$}
  
  \tcp{Compute $\Suc$: First count how often a node is active, then store when} 
  $\sucidx[1 \ldots n{+}1] \gets 0$\;
	\ForEach{$\{u, v\} \in T$}{
    $\sucidx[u] \gets \sucidx[u]+1$; \hspace{1em}
    $\sucidx[v] \gets \sucidx[v]+1$\;
  }
  $\sucbeg[i] \gets 1+\sum_{j=1}^{i-1} \sucidx[j]\ \ \forall 1 \le i \le n{+}1$ \tcp*{Exclusive prefix sum with stop marker}
  copy $\sucidx \gets \sucbeg$\;
  
  Allocate $\Suc[1 \ldots 2\ell]$\;
	\ForEach{$t_i = \{u_i, v_i\} \in T$ for increasing $i$}{
    $\Suc[ \sucidx[u_i] ] \gets i$; \hspace{1em} $\sucidx[u_i] \gets \sucidx[u_i] + 1$\;
    $\Suc[ \sucidx[v_i] ] \gets i$; \hspace{1em} $\sucidx[v_i] \gets \sucidx[v_i] + 1$\;
  }
  reset $\sucidx \gets \sucbeg$\;
  
  $\tau_{v_i} \defrel \text{if }(\sucidx[i] == \sucbeg[i+1]) \text{ then } \infty \text{ else } \Suc[\sucidx[i]]$ \tcp*{Short-hand to read successor}
  
  \vspace{1em}
  \tcp{Fill $\AG$}
  $\adjbeg[i] \gets 1 + \sum_{j=1}^{i-1} \deg(v_j)\ \ \forall 1 \le i \le n{+}1$\tcp*{Exclusive prefix sum with stop marker}
  copy $\adjidx \gets \adjbeg$\;
  Allocate $\AG[1 \ldots 2m]$
  
  \ForEach{$\{a,b\} \in E$}{
    \label{algo:imcb-pushnode}

	$\arraycolsep=1.4pt\begin{array}{ll}
	\lIf{$\tau_a \le \tau_b$} & \text{push $b$ into $\AG(a)$: $\AG[\adjidx[a]] \gets b$; \hspace{1em} $\adjidx[a] \gets \adjidx[a] + 1$} \\
	\lElse{} & \text{push $a$ into $\AG(b)$: $\AG[\adjidx[b]] \gets a$; \hspace{1em} $\adjidx[b] \gets \adjidx[b] + 1$}
	\end{array}$
  }

  \vspace{1em}
  \tcp{Trade}

	\ForEach{trade $t_i = (u, v) \in T$ for increasing $i$}{
		Gather neighbours $\AG(u)$, $\AG(v)$ from $\AG$ using $\adjbeg$\;
    Reset $\adjidx[u] \gets \adjbeg[u]$, $\adjidx[v] \gets \adjbeg[v]$\;
    Advance $\sucidx[u]$ and $\sucidx[v]$, s.t. $\tau_u$ and $\tau_v$ gets next trades\;

    Randomly reassign disjoint neighbours, yielding new neighbours $\Nei{u}$ and $\Nei{v}$.
        
    \ForEach{$(a, b) \in (\{u\} \times \AG'(u)) \cup (\{v\} \times \AG'(v))$}{
      \tcp{Push node edge into $\AG$; same as line \ref{algo:imcb-pushnode}}
		$\arraycolsep=1.4pt\begin{array}{ll}
		\lIf{$\tau_a < \tau_b$} & \text{Push $b$ in $\AG(a)$} \\
		\lElse{} & \text{Push $a$ in $\AG(b)$}
		\end{array}$
		}
	}
	
	\caption{\imcb{} as detailled in section~\ref{subsec:algo-imcb}. }
	\label{algo:imcb}
\end{myalgorithm}

\vfill

\begin{proof}[Proof of Lemma~\ref{lem:imcb-complex}]
  As detailled in Alg.~\ref{algo:imcb}, the computation of $\Suc[\cdot]$ and its auxiliary structures involves scanning over $T$ and $V$ resulting in $\O(n + \ell)$ operations.
  Inserting all edges into $\AG$ requires another $\O(n + m)$ steps.
  
  The $i$-th trade takes $\O(\deg(v_i) + \deg(u_i))$ time to retrieve the input edges and distribute the new states.
  To compute the disjoint neighbours, we insert $\Nei{u_i}$ into a hash set and subsequently issue one existence query for each neighbour in $\Nei{v_i}$; this takes expected time $\O(\deg(v_i) + \deg(u_i))$.
  Since $T$'s constituents are drawn uniformly at random, we estimate the neighbourhood sizes as $\mathbf E[\deg(u_i)] = \mathbf E[\deg(v_i)] = m/n$ yielding the first claim.
  In case of $r$ global trades, $T$ consists of $r$ groups with $n/2$ trades targeting all nodes each.
  Hence, trading requires time
  $
  r \sum_i (\deg(u_i) + \deg(v_i)) = r \sum_{v \in V} \deg(v) = \O(r m).
  $
\end{proof} %
\section{Appendix: EM-GCB}
\label{sec:app-emgcb}
Recall that a global trade can be encoded by a permutation $\pi\colon V \to V$ on the nodes or equivalently on the node indices (see section~\ref{subsection:cb}).
Consequently, generating a uniform random permutation on $[n]$ yields a uniform random global trade.
Injective hash-functions have several computational advantages and can substitute the random permutation:

\begin{mydef}[Relaxed global trade]
	\label{def:pglobaltrade}
	Let $h \colon [n] \to \mathbb N$ be an injective hash-function and $\sequence{a_i}{i=1}{n}$ be the image $\sequence{h(i)}{i=1}{n}$ in sorted order.
	Further let $T_h = \sequence{t_i}{i=1}{n/2}$ where $t_i$ trades the nodes with indices $h^{-1}(a_{2i-1})$ and $h^{-1}(a_{2i})$.
	Hence $h$ implies the global trade $T_h$ analogously to a permutation.
\end{mydef}

In this setting, similar to using permutations, a sequence $T$ of global trades is given by $r$ hash-functions $T = \sequence{h_i}{i=1}{r}$.
Again, \emgcb{} uses the fact that each node $v_i$ is actively traded only once in each round $j$ and can then be addressed by $h_j(i)$ (instead of previously $\pi_j(i)$). %
\section{Linear congruential maps}
\label{sec:lincongmap}

We use linear congruential maps as fast injective hash-functions to model global trades for \empgcb{}.
In this section, some of their useful  properties are shown.
We use the notation $\mathbb Z_p = \{0, 1, \ldots,p-1\}$ and $\mathbb Z_p^* = \{1, \ldots, p-1\}$ for $p$ prime and implicitly use $0 \equiv p \mod p$.
Additionally for a map $h: X \to Y$ we denote the image of $h$ as $\mathbf{im}(h) = \{ h(x) : x \in X \}$.
\begin{mydef}[2-universal hashing]
	Let $H$ be an ensemble of maps from $X$ to $Y$ and $h$ be uniformly drawn from $H$.
	For finite $X$ and $Y$ we call the ensemble $H$ \emph{$2$-universal} if for any two distinct $x_1,x_2 \in X$ and any two $y_1, y_2 \in Y$ and uniform random $h \in H$
	\[ \mathbf P(h(x_1) = y_1 \land h(x_2) = y_2) = |Y|^{-2}. \]
\end{mydef}
\begin{myprop}
	\label{prop:congruentialbijective}
	A linear congruential map $h_{a,b}\colon \mathbb Z_p \to \mathbb Z_p, x \mapsto ax + b \mod p$ for $a \ne 0$ and $p$ prime is a bijection.
\end{myprop}
\begin{proof}
	The translation $\tau_b(x) = x + b \mod p$ and multiplication $\chi_a(x) = ax \mod p$ is injective for all $a \in \mathbb Z_p^*$ and $b \in \mathbb Z_p$.
	Then, the composition $h_{a,b} = (\chi_a \circ \tau_b)$ is also injective and the inverse is given by $h^{-1}_{a,b}(y) = a^{-1}(y-b) \mod p$.
\end{proof}
\begin{mylem}
	\label{lem:statrandom}
	The ensemble $H = \{ h_{a,b}\colon a \in \mathbb Z_p^*,~b \in \mathbb Z_p \}$  is $2$-universal.
\end{mylem}
\begin{proof}
  see Proposition 7 of \cite{Carter1979}.
\end{proof}
The input size will most likely not be prime but linear congruential maps can still be used as injective maps since by the prime number theorem the next larger prime to a number $n$ is on average $\O(\ln(n))$ larger.
Additionally, since $[n]$ is a subset of $\mathbb Z_p$ the $2$-universality also already applies to distinct keys $x_1,x_2 \in [n]$.
The small difference in $n$ and $p$ brings an additional feature we exploit while sending type (ii) messages (see Proposition~\ref{prop:loglogsearch}):
given a lower and upper bound on a hashed value with their respective ranks, one can estimate the rank of an element lying between those bounds.
\begin{mydef}[Sorted rank-map]
	Let $n \in \mathbb N$.
	Further, let $h \colon [n] \to \mathbb N$ be an injective map restricted to $[n]$ and $\pi_h$ be the permutation that sorts $\sequence{h(i)}{i=1}{n}$ ascendingly.
	Denote with $\pi = (h \circ \pi_h) \colon [n] \to \mathbf{im}(h)$ the \emph{sorted rank-map}.
	It is clear that $\pi$ is bijective, and $\pi^{-1}$ remaps a mapped value to its rank in $\mathbf{im}(h)$, see Fig.~\ref{fig:sortedrankmap}.
\end{mydef}
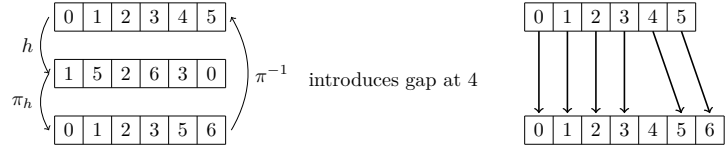
\begin{figure}[t]
	\centering\scalebox{0.75}{
\begin{tikzpicture}[scale=0.5]
	\path (-0.2,-0.5) edge[bend right, ->] node[left] {$h$} (-0.2,-2.5);
	\path (-0.2,-2.5) edge[bend right, ->] node[left] {$\pi_h$} (-0.2,-4.5);
	\path (6.2, -4.5) edge[bend right, ->] node[right] {$\pi^{-1}$} (6.2, -0.5);
	
	\draw (0,0) -- (6,0);
	\draw (0,-1) -- (6,-1);
	\draw (0,0) -- (0,-1);
	\draw (1,0) -- (1,-1);
	\draw (2,0) -- (2,-1);
	\draw (3,0) -- (3,-1);
	\draw (4,0) -- (4,-1);
	\draw (5,0) -- (5,-1);
	\draw (6,0) -- (6,-1);
	\node at (0.5,-0.5) {0};
	\node at (1.5,-0.5) {1};
	\node at (2.5,-0.5) {2};
	\node at (3.5,-0.5) {3};
	\node at (4.5,-0.5) {4};
	\node at (5.5,-0.5) {5};
	
	\draw (0,-2) -- (6,-2);
	\draw (0,-3) -- (6,-3);
	\draw (0,-2) -- (0,-3);
	\draw (1,-2) -- (1,-3);
	\draw (2,-2) -- (2,-3);
	\draw (3,-2) -- (3,-3);
	\draw (4,-2) -- (4,-3);
	\draw (5,-2) -- (5,-3);
	\draw (6,-2) -- (6,-3);
	\node at (0.5,-2.5) {1};
	\node at (1.5,-2.5) {5};
	\node at (2.5,-2.5) {2};
	\node at (3.5,-2.5) {6};
	\node at (4.5,-2.5) {3};
	\node at (5.5,-2.5) {0};
	
	\draw (0,-4) -- (6,-4);
	\draw (0,-5) -- (6,-5);
	\draw (0,-4) -- (0,-5);
	\draw (1,-4) -- (1,-5);
	\draw (2,-4) -- (2,-5);
	\draw (3,-4) -- (3,-5);
	\draw (4,-4) -- (4,-5);
	\draw (5,-4) -- (5,-5);
	\draw (6,-4) -- (6,-5);
	\node at (0.5,-4.5) {0};
	\node at (1.5,-4.5) {1};
	\node at (2.5,-4.5) {2};
	\node at (3.5,-4.5) {3};
	\node at (4.5,-4.5) {5};
	\node at (5.5,-4.5) {6};
\end{tikzpicture} 		\begin{tikzpicture}
		\node at (0,0) {introduces gap at 4\textcolor{white}{xdd}};
		\node at (0, -1) {};
		\end{tikzpicture}
\begin{tikzpicture}[scale = 0.5]
	\draw (0,0) -- (6,0);
	\draw (0,-1) -- (6,-1);
	\draw (0,0) -- (0,-1);
	\draw (1,0) -- (1,-1);
	\draw (2,0) -- (2,-1);
	\draw (3,0) -- (3,-1);
	\draw (4,0) -- (4,-1);
	\draw (5,0) -- (5,-1);
	\draw (6,0) -- (6,-1);
	\node at (0.5,-0.5) {0};
	\node at (1.5,-0.5) {1};
	\node at (2.5,-0.5) {2};
	\node at (3.5,-0.5) {3};
	\node at (4.5,-0.5) {4};
	\node at (5.5,-0.5) {5};
	
	\draw (0,-4) -- (7,-4);
	\draw (0,-5) -- (7,-5);
	\draw (0,-4) -- (0,-5);
	\draw (1,-4) -- (1,-5);
	\draw (2,-4) -- (2,-5);
	\draw (3,-4) -- (3,-5);
	\draw (4,-4) -- (4,-5);
	\draw (5,-4) -- (5,-5);
	\draw (6,-4) -- (6,-5);
	\draw (7,-4) -- (7,-5);
	\node at (0.5,-4.5) {0};
	\node at (1.5,-4.5) {1};
	\node at (2.5,-4.5) {2};
	\node at (3.5,-4.5) {3};
	\node at (4.5,-4.5) {4};
	\node at (5.5,-4.5) {5};
	\node at (6.5,-4.5) {6};
	
	\draw[->, thick] (0.5, -1) -- (0.5, -3.9);
	\draw[->, thick] (1.5, -1) -- (1.5, -3.9);
	\draw[->, thick] (2.5, -1) -- (2.5, -3.9);
	\draw[->, thick] (3.5, -1) -- (3.5, -3.9);
	\draw[->, thick] (4.5, -1) -- (5.5, -3.9);
	\draw[->, thick] (5.5, -1) -- (6.5, -3.9);
\end{tikzpicture} %
}
	\caption{The sorted rank-map for $n = 6$ and $h\colon [n] \to \mathbb Z_7, x \mapsto 4x + 1$.
		For the set $\{0, 1, 2, 3\}$ the sorted rank-map $\pi$ is just the identity.
		In contrast for $x \in \{4, 5\}$ the value $x$ is mapped to $\pi(x) = x + 1$.}
	\label{fig:sortedrankmap}
\end{figure}
\begin{myrem}
	The sorted rank-map $\pi$ can only shift the original values and is thus monotonically increasing, see Fig.~\ref{fig:sortedrankmap}.
	The shift in value is given by $\pi(x) - x$ and is monotonically increasing, too.
	By applying $\pi$ we introduce gaps in the set $\mathbb Z_p$ from $[n]$, refer to Fig.~\ref{fig:sortedrankmap}.
\end{myrem}
\begin{myprop}
	\label{prop:bounds}
	Let $n \in \mathbb N$ and $p \ge n$ be a prime number.
	Further, let $h\colon [n] \to \mathbb Z_p$ be a linear congruential map and $\pi$ be its sorted rank-map.
	If we want to compute the rank of $y \in \mathbf{im}(h)$ and know $x, x' \in [n]$ where $h(x) \le y \le h(x')$ then we can bound the rank $\pi^{-1}(y)$ of $y$ by using the shifts of $x$ and $x'$: 
	$y - (\pi(x') - x') \le \pi^{-1}(y) \le y - (\pi(x) - x).$
\end{myprop}
\begin{proof}
	The sorted rank-map $\pi$ is by definition monotone increasing, see also Fig.~\ref{fig:sortedrankmap}.
	It follows that $\pi(x) = x + k$, $\pi(x') = x' + k'$ and $k \le k'$ for some $k, k' \in \mathbb N$.
	By monotonicity $\pi(\pi^{-1}(y)) = \pi^{-1}(y) + s$ for $s \in \{k, \ldots, k'\}$, resulting in inequalities
	\begin{align*}
	\pi^{-1}(y) + k \le & \ y,  \\
	& \ y \le \pi^{-1}(y) + k'.
	\end{align*}
	By subtracting $k$ and $k'$ on both sides, the claim follows.
\end{proof}
With Proposition~\ref{prop:bounds} we can reduce the number of candidates to search in.
This is especially useful, when working on a smaller contiguous part of the data (see \empgcb{},  section~\ref{subsec:algo-empgcb}).
\begin{myexp}
	Let $n$ and $h$ be given from Fig.~\ref{fig:sortedrankmap}.
	It is clear that the hashed-values are given by $\mathbf{im}(h) = \{0, 1, 2, 3, 5, 6\}$.
	Suppose the rank of $2$ in $\mathbf{im}(h)$ has to be computed given the outer values e.g. that $\pi(0) = 0$ and $\pi(5) = 6$.
	Then by Propositon~\ref{prop:bounds}
	\begin{align*}
		2 - (\pi(5) - 5) &\le \pi^{-1}(2) \le 2 - (\pi(0) - 0), \\
		1 &\le \pi^{-1}(2) \le 2.
	\end{align*}
	Thus, the rank of $2$ in $\mathbf{im}(h)$ is either $1$ or $2$.
\end{myexp} %
\section{Appendix: \empgcbHead{}}
\label{sec:app-empgcb}
\empgcb{} achieves parallelism by performing multiple trades concurrently.
In contrast to \emgcb{}, rather than only retrieving the first two necessary adjacency rows for the single next trade, a whole chunk of data is loaded and maintained in \imcb{}'s adjacency list to store neighbours for a subset of nodes.
The adjacency list is further used as a way to transport messages within a loaded macrochunk.
Observe that at most $2m$ many messages are sent in a global trade round since only neighbourhood information is forwarded.

The idea is to split the messages into chunks of size $\mathcal M = cM$ where $c \in (0, 1)$ which can be processed in IM.
For this, \empgcb{} loads and proceesses all messages targetted to the next $n/k$ nodes for a constant $k$ and performs the corresponding trades concurrently.
This subdivides the messages and its processing into $k$ \emph{macrochunks}.
If a macrochunk is too large, it cannot be fully kept in IM resulting in unstructured I/O in the trading process.
The choice of $k$ should therefore additionally consider the variance.
An analysis on the size of the macrochunks is given in section~\ref{section:analysis}.

\subsection{Data structure for message transportation}

Recall in subsection~\ref{subsec:algo-empgcb} that each macrochunk is subdivided into many \emph{microchunks} and processed in batches.
During the trading process \empgcb{} has to differentiate between messages that are sent (i) within a microchunk, (ii) between microchunks of the same batch (iii) and microchunks processed later.
To support both type~(i) and type~(ii) messages we organise the messages of the current macrochunk in an adjacency vector data structure similar to \imcb{}.
Instead of forwarding these messages in a TFP-fashion, \empgcb{} inserts them directly into the adjacency data structure.
We rebuild the data structure for each macrochunk requiring the degrees of the $n/k$ loaded nodes to leave gaps if messages are missing.
In a preprocessing step we provide \empgcb{} with this information by inserting messages $\langle h_r(v), \deg(v), v \rangle$ into a separate priority queue.
Initialising the adjacency vector can now be done by loading the degrees for the next $n/k$ targets and reserving for each target $h_r(v)$ the necessary $\deg(v)$ slots.
Messages $\langle r, h_r(v), x \rangle$ targetted to the node $v$ can then be inserted in an unstructured fashion in IM.
This can be done in parallel for all targets in the macrochunk:
first the retrieved messages are sorted in parallel and then accessed concurrently after determining delimiters by a parallel prefix sum over the message counts.

For a trade $t = \{u_i, v_i\}$ of targets $h_r(u_i)$ and $h_r(v_i)$ the assigned processor can determine if the $t$ is tradable by checking whether $\deg(u_i)$ and $\deg(v_i)$ match the number of available messages.
After performing the trade, we forward the updated adjacency information.
Assume that the edge $\{u_i, x\}$ has to be send to a later trade in the same global trade.
\begin{enumerate}[(i)]
	\item 
		If $x$ is traded within the processed microchunk there is no synchronisation required and $u_i$ can be inserted into the row corresponding to target $h_r(x)$.
	\item 
		If $x$ is traded within the currently processed batch the processor has to insert $u_i$ into the row corresponding to target $h_r(x)$ with synchronisation.
		This yields a data dependency in the parallel execution.
		Inferring if the trade for $x$ belongs to the current batch can be done by comparing $h_r(x)$ to the maximum target of the batch.
	\item 
		If $x$ is traded in a later microchunk, it either belongs to the same macrochunk or a later one (of the same global trade).
		For the former \empgcb{} proceeds similar to type (ii) without processing foreign trades.
		In the latter case \empgcb{} inserts a message $\langle r,  h_r(x), u_i \rangle$ into the priority queue.
\end{enumerate}
Addressing the adjacency row of a target $h_r(u)$ can be done by computing the rank of $h_r(u)$ in the retrieved $n/k$ targets.
Since the separate priority queue provides all loaded targets by messages $\langle h_r(u), \deg(u), u \rangle$, we can perform a binary search and obtain the rank in time $\O(\log(n/k))$.

For linear congruential maps (section~\ref{sec:lincongmap}) we can do better:
\begin{myprop}\label{prop:loglogsearch}
	Let $h$ be a linear congruential map.
	Then, heuristically computing the row (rank) corresponding to $h(u)$ requires $\O(\log \log n)$ time.
\end{myprop}
\begin{proof}
	The next larger prime $p$ to $n$ is heuristically $\ln (n)$ larger than $n$.
	After loading all messages $\langle h(u), \deg(u), u \rangle$ for the current macrochunk the smallest and largest hashed value of the current macrochunk are known.
	By subtracting both values by the already processed number of targets and using Proposition~\ref{prop:bounds} the search space can be reduced to $\O(\log n)$ elements.
  Application of a binary search on the remaining elements yields the claim.
\end{proof}
As already mentioned, if a trade has not received all its required messages, the assigned processor cannot perform the trade yet and therefore skips it.
This can only happen within a batch when type (ii) messages occur.
In section~\ref{section:analysis} we argue that this happens rarely.
The processor that inserts the last message for that particular trade will perform it instead.

\subsection{Improvements for type (iii) messages}
Messages inserted into the priority queue need to contain the round-id to process global trades separately.
Observe however that in a sequence of global trades, messages are only send to the current and subsequent round.
We therefore modify our data structure, omitting the round from \emph{every} message reducing the memory footprint significantly.
Recall that, as an optimisation for $m = \O(M^2/B)$ edges, \empgcb{} uses external memory queues for each of the $k$ macrochunks of both global trade rounds.

A previously generated message $\langle r, h_r(u), x \rangle$ is now inserted into the corresponding queue containing messages for $h_r(u)$.
Again, in a preprocessing step \empgcb{} determines for each queue its target range.
For this, the separate priority queue containing messages $\langle h_r(u), \deg(u), u \rangle$ is read while extracting every $(n/k)$-th target (retrieving every element results in a sequence of sorted messages).
This enables the computation of the correct queue for $h_r(u)$ with a binary search in time $\O(\log(k))$.
Naturally since both the current and subsequent round are relevant, \empgcb{} employs $k$ external memory queues for each.
If a global trade is finished, the $k$ EM queues of the currently processed and finished round can be reused for the next global trade.
\empgcb{}'s pseudo code can be found in Algorithm~\ref{algo:curveball-em-par}.
\begin{myalgorithm}[t]
	\SetKwFor{Pardo}{pardo}{}
	
	\SetKwData{auxInfoToTarget}{AuxInfoToTarget}
	
	\KwData{Trade sequence $T = \sequence{h_i}{i=1}{r}$, simple graph $G=(V, E)$ as edge list $E$} 
	\KwResult{Randomised graph $G'$}
	\tcp{Initialisation: provide auxiliary info and initialise with edges}
	
	\ForEach{node $u \in V$}{
		Send $\langle h_1(u), \deg(u), u \rangle$ via \auxInfoToTarget 
		\tcp{Send node and degree to target}
	}
	Sort \auxInfoToTarget lexicographically \\
	Scan \auxInfoToTarget and determine bounds for the $k$ queues
	
	\ForEach{edge $e = [u,v]$ in $E$}{
		Insert $e$ according to $h_1$ into one of the corresponding queues
	}

	\vspace{0.5em}
	\tcp{Execution: Process rounds and macrochunks}

	\For{round $R = 1, \ldots, r$}{
		\For{macrochunk $K = 1, \ldots, k$}{
			Retrieve auxiliary data $\langle h_R(u), \deg(u), u \rangle$ from \auxInfoToTarget
			
			Load and sort messages of the $K$-th queue
			
			Insert the messages into the adjacency list $\AG$ in parallel
			
			\For{batch $\mathcal B = 1, \ldots, z$}{
				\Pardo {the $i$-th processor works on the $i$-th microchunk of batch $\mathcal B$} {
					\For{a trade $t = \{u, v\}$}{
						Retrieve $A_u$ and $A_v$ from $\AG$
						
						With $\deg(u)$ and $\deg(v)$ determine whether tradable
						
						\uIf{tradable}{
							Compute $A_u'$ and $A'_v$
							
							Forward each resulting edge \newline
							 worksteal if inserted message fills all necessary data
						}
						\lElse {
							Skip
						}
					}
				}
			}
		}
		\If{$R < r$}{
			Clear \auxInfoToTarget and refill for $h_{R+1}$ (repeat steps 3 to 5)
		}
	}
	\caption{\empgcb{} as detailled in section~\ref{subsec:algo-empgcb} and section~\ref{sec:app-empgcb}.}
	\label{algo:curveball-em-par}
\end{myalgorithm}

\section{Analysis of \empgcb}
\label{section:analysis}

\subsection{Macrochunk size}
As already mentioned, the number of incoming messages may exceed the size of the internal memory $M$, since we partition the nodes into chunks which then may receive a different number of messages.
Therefore some analysis on the size of the maximum macrochunk is necessary.
Denote with $\mathcal N(\mu, \sigma^2)$ the distribution of a Gaussian r.v.\ with mean $\mu$ and variance $\sigma^2$.
A macrochunk holds the sum of $n/k$ many iid degrees and is thus approximately Gaussian with mean $2m/k$ and variance $n/k \cdot \mathbf{Var}(D)$ where $D$ is distributed to the underlying degree distribution.
This approximation gets better for larger values of $n/k$ and is thus a suitable approximation for large graphs.
Denote with $S_1, \ldots, S_k$ the sizes of all $k$ macrochunks.

When determining a suitable choice of $k$, it is necessary to consider both the mean and the variance of the maximum macrochunk $\max_{1\le i \le k}S_i$.
The largest macrochunk may receive many high-degree nodes exceeding the size of the internal memory $M$.
We thus bound its number in Corollary~\ref{cor:expectation} and Corollary~\ref{cor:variance}.

\begin{mylem}
	\label{lem:variancemaxnormal}
	Let $Y = \max_{1 \le i \le k} X_i$, where the $X_i$ are iid r.v.\ distributed as $\mathcal N(0, \sigma^2)$.
	Then,
	$ \mathbf E[Y] \le \sigma \sqrt{2\log(k)}. $
\end{mylem}
\begin{proof}
	The following chain of inequalities holds $e^{t\mathbf E[Y]} \le \mathbf E[e^{tY}] = \mathbf E[\max_{1 \le i \le k}e^{tX_i}] \le \sum_{i=1}^{k}\mathbf E[e^{tX_i}] = k e^{t^2\sigma^2/2}$, where in order Jensen's inequality\footnote{For a convex function $f$ and non-negative $\lambda_i$ with $\sum_{i=1}^{n}\lambda_i = 1$ follows $f(\sum_{i=1}^{n}\lambda_i x_i) \le \sum_{i=1}^{n}\lambda_if(x_i)$.} monotonicity and non-negativity of the exponential function as well as the definition of the moment generating function of a Gaussian r.v.\ have been applied.
	Taking the natural logarithm and dividing by $t$ on both sides (ruling out $t \ne 0$) yields $\mathbf E[Y] \le \frac{\log(k)}{t} + \frac{t \sigma^2}{2}$, which is minimized by $t = \sqrt{2\log(k)}/\sigma$.
	The above proof is a special case in a proof of \cite{Massart2007}.
\end{proof}
\begin{mycor}
	Let $Y = \max_{1 \le i \le k}S_i$.
	By approximating $S_i$ with a Gaussian r.v.\ $N_i$ with $\mu = \mathbf E[S_i]$ and $\sigma^2 = \mathbf{Var}(S_i)$, one gets an approximate upper bound on $Y$:
	\[ \mathbf E[Y] \approx \mathbf E\left[\max_{1 \le i \le k}N_i\right] \le \mathbf E[S_1] + \sqrt{2 \log(k) \mathbf{Var}(S_1)} = \mathbf E[S_1] + \sqrt{\frac{n\log(k)}{2k} \mathbf{Var}(D)}.\]
	\label{cor:expectation}
\end{mycor}
\begin{proof}
	Since $\max_{1 \le i \le k}N_i$ is centred around $\mu$, it is identically distributed to $\mu + \max_{1 \le i \le k}N_i'$ where $N_i'$ has the same variance but is centred around $0$.
	By applying Lemma~\ref{lem:variancemaxnormal} to $\max_{1 \le i \le k}N_i'$ the claim follows, since $\mathbf E\left[\max_{1 \le i \le k}N_i\right] =  \mu + \mathbf E\left[\max_{1 \le i \le k}N_i'\right]$.
\end{proof}
\begin{mylem}
	\label{lem:maximumvariance}
	Let $X_1, \ldots, X_k$ be iid and $Y = \max_{1 \le i \le k}X_i$.
	Then, 
	$ \mathbf{Var}(Y) \le k\mathbf{Var}(X_1).$
\end{mylem}
\begin{proof}
	For $Z, Z'$ iid. $\mathbf E[(Z - Z')^2] = 2 \mathbf{Var}(Z)$ holds, since $\mathbf E[Z^2 - 2ZZ' + Z'^2] = 2\mathbf E[Z^2] - 2\mathbf E[Z]^2$.
	Now, let $Y' = \max_{1 \le i \le k} X_i'$ be an independent copy of $Y$ and $r > 0$. 
	
	First, the inequality $\mathbf P(|Y - Y'|^2 > r) \le \sum_{i=1}^{k}\mathbf P(|X_i - X_i'|^2 > r)$ is shown.
	We show the implication that when $|Y-Y'|^2 > r$ then there exists an index $i$ s.t.\ $|X_i - X_i'|^2 > r$.
	If $|Y-Y'|^2 > r$ holds, then w.l.o.g.\ let $Y = X_i$ and $Y' = X'_j$ and $Y > Y'$ s.t.\ $|X_i - X'_j|^2 > r$.
	By maximality the following chain of inequalities holds $X_i > X_j' \ge X_i'$.
	Which already implies $|X_i > X_i'| > r$ and consequently $\mathbf P(|Y - Y'|^2 > r) \le \mathbf P(\text{exists index $i$ s.t.\ } |X_i - X_i'| > r)$.
	
	Now by bounding the union, one gets $\mathbf P(|Y-Y'|^2 > r) \le \sum_{i=1}^{k}\mathbf P(|X_i - X_i'|^2 > r)$.
	At last, integrating $r$ from $0$ to $\infty$ yields $2\mathbf{Var}(Y) = \mathbf E[(Y-Y')^2] \le k \mathbf E[(X_1 - X_1')^2]= 2k\mathbf{Var}(X_1)$, which concludes the proof.
\end{proof}
\begin{mycor}
	\label{cor:variance}
	Let $Y = \max_{1 \le i \le k}S_i$.
	Then,
	$ \mathbf{Var}(Y) \le k\mathbf{Var}(S_1) = n\mathbf{Var}(D)$.
\end{mycor}
\begin{proof}
	This is a special case of Lemma~\ref{lem:maximumvariance}.
\end{proof}

The probability mass of a Gaussian r.v.\ is concentrated around its mean, e.g.\ the tails vanish very quickly, see Proposition~\ref{prop:tailsgaussian}.
This heuristically additionally holds true for the maximum macrochunk size (Lemma~\ref{lem:concentrationinequalitynormal}).

\begin{myprop}
	\label{prop:tailsgaussian}
	Let $X$ be a standard Gaussian r.v.\ and $f(x) = \frac{1}{\sqrt{2\pi}}e^{-x^2/2}$ be its probability density function.
	Let $t > 0$ then it holds
	$ \mathbf P(X > t) \le \exp(-t^2/2) / \sqrt{2\pi} / t = \O\left(\frac{e^{-t^2/2}}{t}\right). $
\end{myprop}
\begin{proof}
	The value of $\mathbf P(X > t)$ equals $\int_t^\infty \frac{1}{\sqrt{2\pi}}e^{-x^2/2}dx$.
	Since the integrating variable ranges from $[t, \infty)$ then $\frac{x}{t} \ge 1$ s.t.\ $\mathbf P(X > t) \le \int_{t}^{\infty}\frac{x}{t}\frac{1}{\sqrt{2\pi}}e^{-x^2/2}dx = \frac{1}{t}\frac{e^{-t^2/2}}{\sqrt{2\pi}}$.
\end{proof}

\begin{mylem}
	\label{lem:concentrationinequalitynormal}
	Let $Y = \max_{1 \le i \le k} N_i$ where $N_i$ are iid standard Gaussian random variables.
	Then $
	 \mathbf P(Y > t) = \mathcal O \left( k \exp(-t^2 / 2) / t \right). $
\end{mylem}
\begin{proof}
	The claim follows by the following calculation:
	\[
	\mathbf P(Y > t) = \mathbf P\left(\max_{1 \le i \le k}N_i > t\right) = \mathbf P\left(\exists\,i \text{ s.t. } N_i > t \right)
	\le \sum_{i=1}^{k}\mathbf P(N_i > t) 
	= \O\left(k \cdot \frac{e^{-t^2/2}}{t}\right). \]

	If for any random variable $N_i {>} t$, then already $\max_{1 \le i \le k}N_i {>} t$, inversely if $\max_{1 \le i \le k} N_i {>} t$ then there exists a $N_i$ s.t.\ $N_i > t$, which shows the first equality.
	After applying the union bound and Proposition~\ref{prop:tailsgaussian} the claim follows.
\end{proof}

\subsection{Heuristic on intra-batch dependencies}
In \empgcb, if information on an edge $\{u, w\}$ has to be inserted into the same batch a dependency arises.
We will now argue that this happens not too often when the number of batches $z$ is chosen sufficiently large.

\begin{mylem}
	\label{lem:dependencies}
	Let $\mathcal B$ be the set of targets for a batch.
	Assuming uniform neighbours, the number of dependencies from $\mathcal B$ to $\mathcal B$ heuristically is
	$  \binom{p}{2}\frac{2m}{k^2z^2p^2}. $
\end{mylem}
\begin{proof}
	By construction $|\mathcal B| = \frac{n}{kz}$ since $\mathcal B$ is part of an equal subdivision of a macrochunk.
	Each individual microchunk consists of $\frac{n}{kzp}$ many targets for the same reason.
	The $i$-th microchunk therefore has $(p - i)\frac{n}{kzp}$ many critical targets.
	On average each microchunk generates $\deg_{}\frac{n}{kzp} = \frac{2m}{kzp}$ many messages that need to be forwarded.
	For an edge produced by the $i$-th microchunk assume uniformity on the neighbours $A$, then $V_i$ is the number of critical messages where
	$ V_i = \sum_{i=1}^{n} 1_{i \in A} 1_{i \in h^{-1}(\mathcal B)}. $
	Its expectation is given by
	\[ \mathbf E[V_i] = \sum_{i=1}^{n}\mathbf P(i \in A)\mathbf P(i \in h^{-1}(\mathcal B)) = n\frac{\deg_{\avg}}{n}\frac{\frac{n(p-i)}{kzp}}{n} = \deg_{\avg}\frac{p-i}{kzp}. \]
	Now let the total number of messages from the $i$-th microchunk to $\mathcal B$ be $H_i$. 
	Since each microchunk holds $\frac{n}{kzp}$ many nodes, $H_i$ is given by
	\[ \mathbf E[H_i] = \frac{n}{kzp}\mathbf E[V_i] = \frac{2m(p-i)}{k^2z^2p^2}.  \]
	By summing over all $p$ microchunks, e.g.\ $\sum_{i=1}^{p}\mathbf E[H_i]$ the claim follows.
\end{proof}

\begin{myexp}
	\label{exp:dependencies}
	Consider Lemma~\ref{lem:dependencies} where $m = 12 \times 10^9$, $k = 32$, $z = 2^{11}$ and $p = 16$.
	The average number of messages in the batch is given by $m/kz \ge 1.8 \times 10^5$.
	And Lemma~\ref{lem:dependencies} predicts a count of less than $4$ critical messages on average in a batch.
\end{myexp}

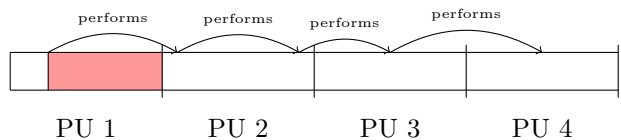
\begin{figure}
	\centering
	\begin{tikzpicture}
	\def\h{0.5}
	\def\a{0}
	\fill[red!40] (0.5,\a) -- (0.5, \a + \h) -- (2, \a + \h) -- (2, \a);
	\draw (0,\a) -- (8, \a);
	\draw (0,\a + \h) -- (8, \a +\h);
	\draw (0,\a) -- (0, \a + \h);
	\draw (0.5,\a) -- (0.5, \a + \h);
	\path[->] (0.5, \a+\h) edge [bend left=30] node[above] {\tiny performs} (2.2, \a + \h);
	\path[->] (2.2, \a+\h) edge [bend left=30] node[above] {\tiny performs} (3.8, \a + \h);
	\path[->] (3.8, \a+\h) edge [bend left=30] node[above] {\tiny performs} (5, \a + \h);
	\path[->] (5, \a+\h) edge [bend left=30] node[above] {\tiny performs} (7, \a + \h);
	\draw (2,\a-0.1) -- (2, \a + \h+0.1);
	\draw (4,\a-0.1) -- (4, \a + \h+0.1);
	\draw (6,\a-0.1) -- (6, \a + \h+0.1);
	\draw (8,\a-0.1) -- (8, \a + \h+0.1);
	\node (pu1) at (1, \a - 0.5) {PU 1};
	\node (pu1) at (3, \a - 0.5) {PU 2};
	\node (pu1) at (5, \a - 0.5) {PU 3};
	\node (pu1) at (7, \a - 0.5) {PU 4};
	\end{tikzpicture}
	\caption{The arrows represent the long chain of trades that are getting work-stolen from the first PU where $p = 4$. The red marked area represents still untouched trades of the first microchunk that will get processed \emph{after} the long chain by the first PU.}
	\label{fig:workstealchain}
\end{figure}

Theoretically by Lemma~\ref{lem:dependencies} the number of critical messages is very small if $z$ is set to be sufficiently large.
Therefore waiting and stalling for missing messages is inefficient and should be avoided.
\empgcb{} thus skips a trade when it cannot be performed and is later executed by the processor that adds the last missing neighbour.
However, since a work-stealing processor spends time on a trade that is possibly assigned to another microchunk, it is not working on its own.
Therefore messages coming from that particular microchunk are generated later down the line.
This may be especially bad when a PU performs a chain of trades that it was not originally assigned to as illustrated in Fig.~\ref{fig:workstealchain}.
Since work-stealing can only be done in a time-forward fashion, the chain length therefore is geometrically distributed (in fact, the probability declines in each step since less targets are critical) and is thus whp of order $\O(1)$ by Proposition~\ref{prop:geomwhp}.

\begin{myprop}
	\label{prop:geomwhp}
	Let $X$ be geometrically distributed with parameter $(1 - 1/z^2)$  for $z > 1$.
	Then,
	$\mathbf P(X > t) = \frac{1}{z^{2t}} = e^{-2\ln(z)t}.$
\end{myprop}
\begin{proof}
	The claim follows by $\mathbf P(X > t) = 1/z^{2t}$ and setting $t = \O(1)$.
\end{proof}
\section{Additional experimental results}
\subsection{Swaps performed by Curveball and Global Curveball}
\label{sec:app-add-experiments} \label{sec:experiment-compare-global-trades-to-trades}

In Fig.~\ref{fig:comparison-curveball} we counted the number of neighbourhood swaps in $n/2$ uniform trades and a single global trade and obtain the fraction of performed swaps to all possible swaps.
These experiments are performed on a series of $10$-regular graphs and powerlaw graphs with increasing maximum degree.
Both algorithms perform a similar count of swaps and suggest no systematic difference.
As expected, for regular graphs the fraction of performed swaps goes to $1/2$ for an increasing number of nodes, since with increasing $n$ the number of common neighbours goes to zero.
On the other hand the fraction of performed swaps decreases for powerlaw graphs with a higher maximum degree.

\begin{figure}[t]
	\begin{center}  
		\vspace{-1em}
		\scalebox{0.42}{\includegraphics{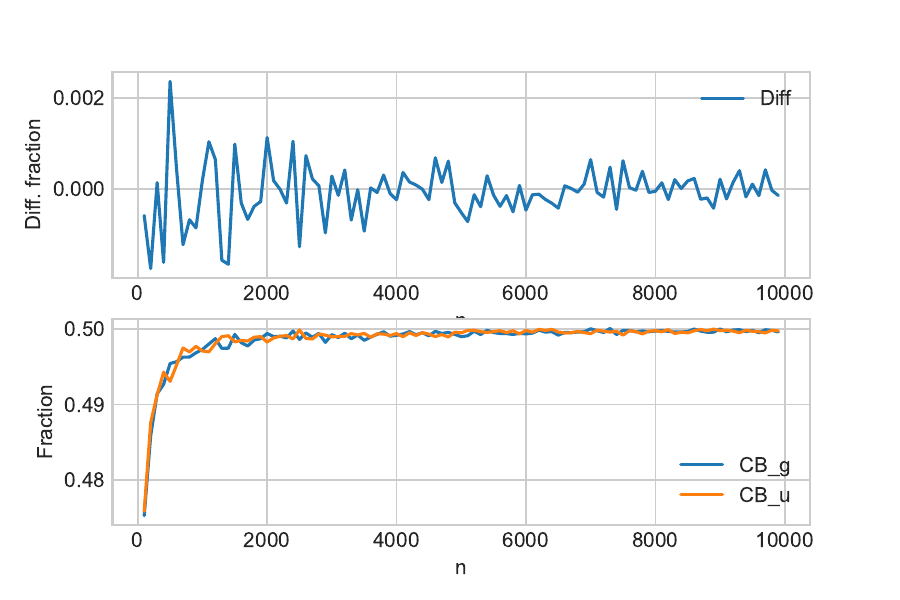}}
		\scalebox{0.42}{\includegraphics{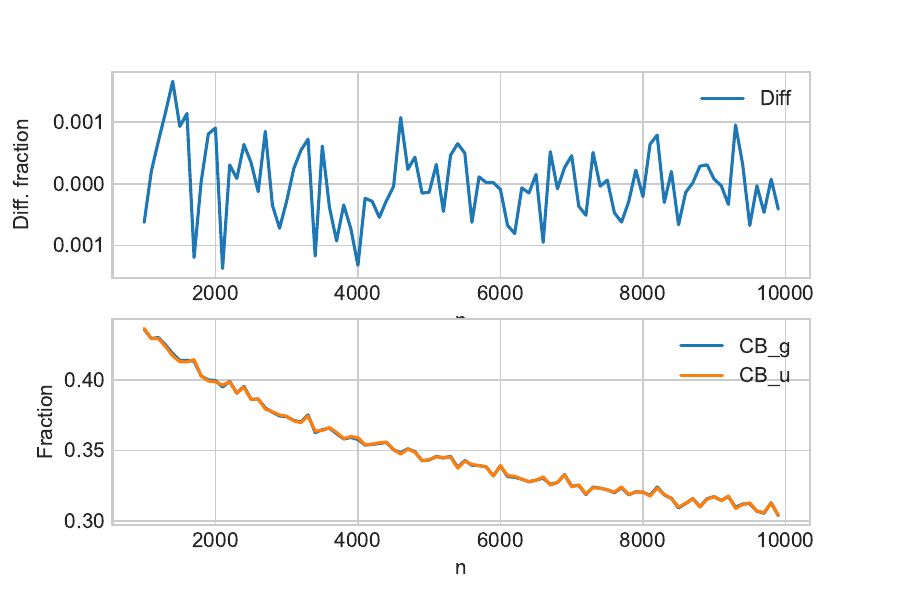}}
		\vspace{-1.5em}
	\end{center}
	\caption{
		The average fraction of performed neighbourhood swaps of $n/2$ uniform trades and a single global trade.
		\textbf{Left:} $10$-regular graphs for increasing $n$.
		\textbf{Right:} powerlaw graphs realised from $\pld{10}{n/20}{2}$ for increasing $n$ by the Havel-Hakimi algorithm.
	}
	\label{fig:comparison-curveball}
\end{figure}

\subsection{Autocorrelation time of Curveball and Edge Switching}

\label{subsec:thinkpen-results}
\begin{figure}[h]
	\centering
	\includegraphics[scale=0.43]{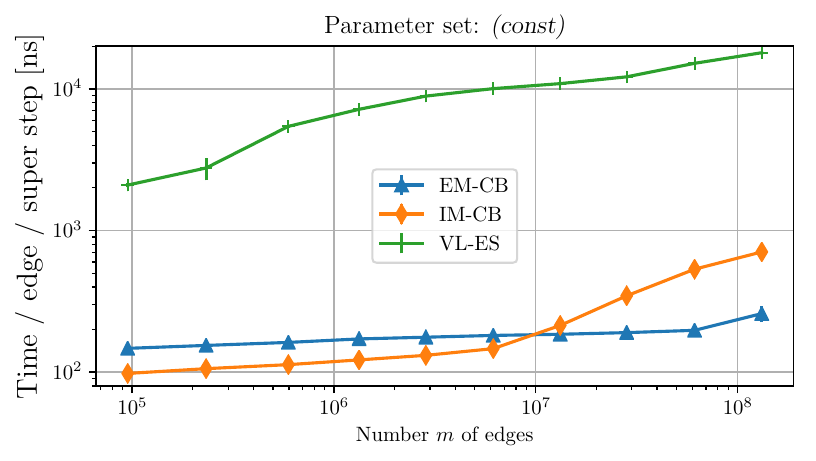}\hfill
	\includegraphics[scale=0.43]{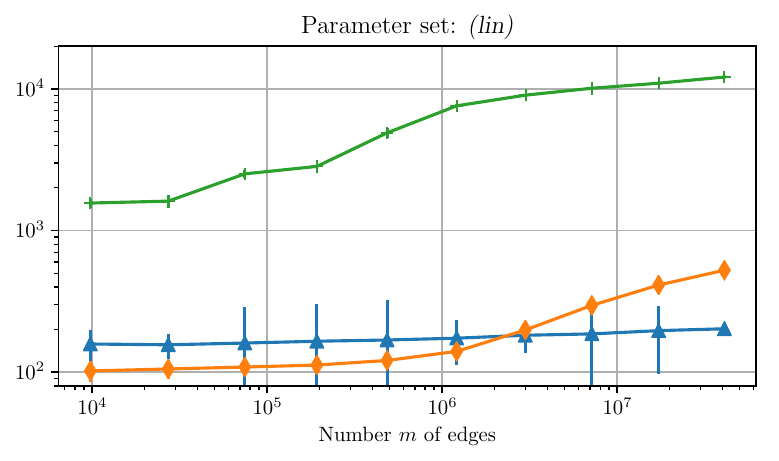}
	\vspace{-0.5em}
	\caption{
		Runtime per edge and super step of \imcb{} and \emcb{} compared to state-of-the-art IM edge switching \vles{}. 
		Each data point is the median of $S \ge 5$ runs over 10~super steps each.    
		The left plot contains the (const)-parameter set, the right one (linear).
		Machine: Intel i7-6700HQ CPU (4 cores), 64~GB RAM, Ubuntu Linux 17.10 with kernel~4.13.0-38.
	}
	\label{fig:runtimes-thinkpen}
\end{figure}

\begin{figure}[h]
	\centering
	\includegraphics[scale=0.45]{independent_edges_pld_25}
	\includegraphics[scale=0.45]{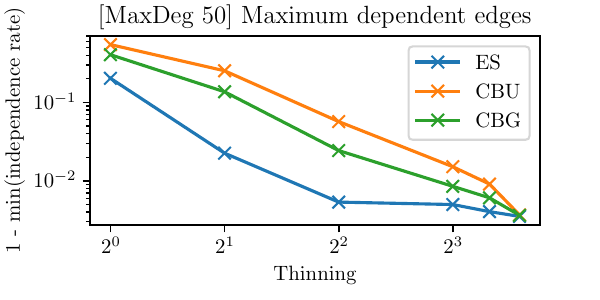}
	\includegraphics[scale=0.45]{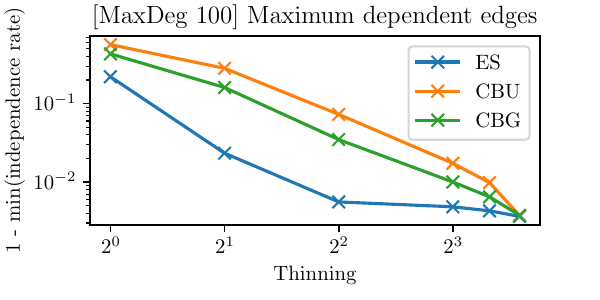}
	\includegraphics[scale=0.45]{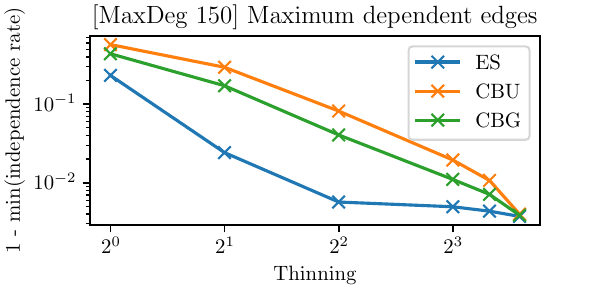}
	\includegraphics[scale=0.45]{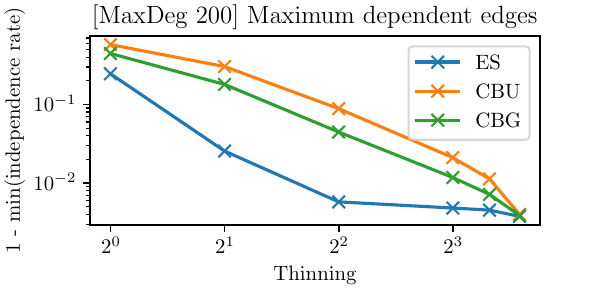}
	\includegraphics[scale=0.45]{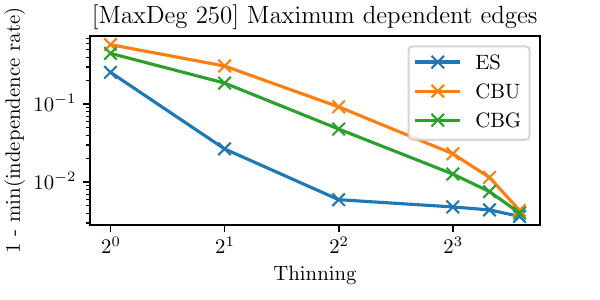}
	\includegraphics[scale=0.45]{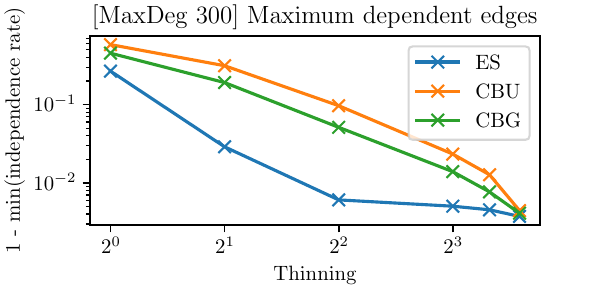}
	\includegraphics[scale=0.45]{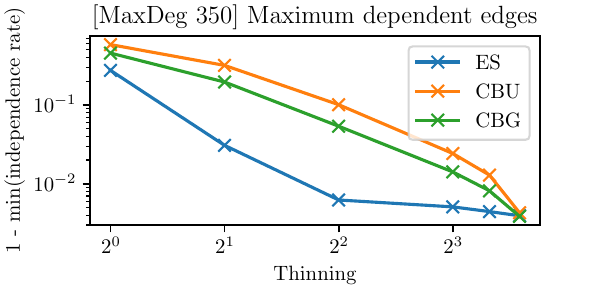}
	\includegraphics[scale=0.45]{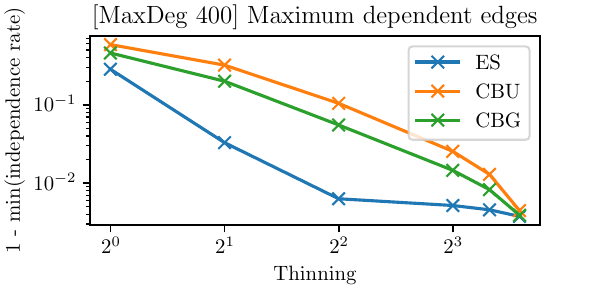}
	\includegraphics[scale=0.45]{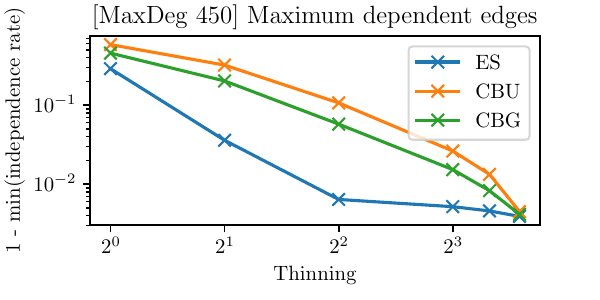}
	\includegraphics[scale=0.45]{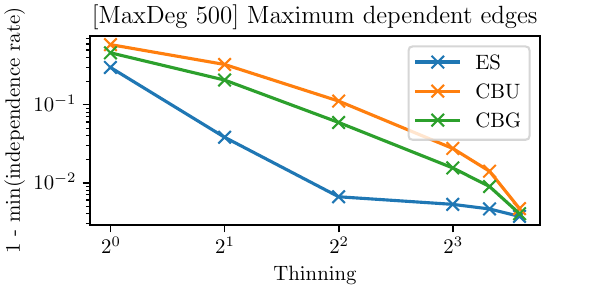}
	\includegraphics[scale=0.45]{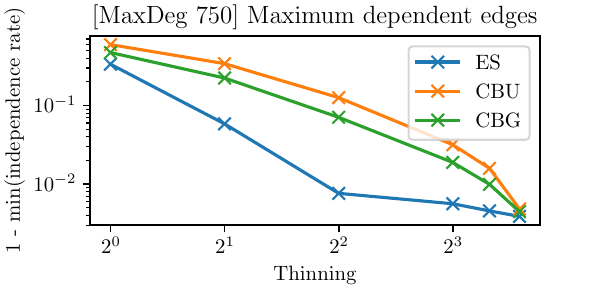}
  \vspace{-0.5em}
	\caption{
		 Fraction of edges still correlated as function of the thinning parameter $k$ for graphs with $n = 2 {\cdot} 10^3$ nodes and degree distribution $\pld{a}{b}{\gamma}$ with $\gamma = 2$, $a = 5$, and several different values for $b$.
		 The (not thinned) long Markov chains contain 6000 super steps each.
	}
	\label{fig:independencerates}
\end{figure}

\end{document}